\documentclass[11pt]{article}

\usepackage[centertags]{amsmath}
\usepackage{amsfonts}
\usepackage{amssymb}
\usepackage{amsthm}
\usepackage{newlfont}
\usepackage{epsfig}
\usepackage{amscd}

\newcommand{\RR}{{\mathbb R}}

\newcommand{\CC}{{\mathbb C}}
\newcommand{\beq}{\begin{equation}}
\newcommand{\eeq}{\end{equation}}
\newcommand{\ba}{\begin{array}}
\newcommand{\ea}{\end{array}}
\newcommand{\bea}{\begin{eqnarray}}
\newcommand{\eea}{\end{eqnarray}}

\usepackage{graphicx}
\usepackage{color}
\newcommand{\corr}[1]{{\color{black}{#1}}}
\usepackage{transparent} 

\newtheorem{theorem}{Theorem} 

\newtheorem{lemma}{Lemma}[section]
\newtheorem{proposition}{Proposition}[section]
\newtheorem{definition}{Definition} 
\newtheorem{example}{Example}[section]

\numberwithin{equation}{section}
\allowdisplaybreaks

\begin{document}

\begin{center}
{\large   \bf The direct scattering problem for the perturbed $\textrm{Gr}(1, 2)_{\ge 0}$ Kadomtsev-Petviashvili  solitons} 
 
\vskip 15pt

{\large  Derchyi Wu}

\vskip 5pt

{ Institute of Mathematics, Academia Sinica, 
Taipei, Taiwan}

e-mail: {\tt mawudc@gate.sinica.edu.tw}

\vskip 10pt

{\today}

\end{center}

\begin{abstract}
Regular Kadomtsev-Petviashvili (KP)  solitons have been investigated and classified successfully by the Grassmannian. We provide rigorous analysis for the direct scattering problem of perturbed   $\textrm{Gr}(1, 2)_{\ge 0}$ KP solitons. 
\end{abstract}

\section{Introduction}\label{S:motivation}
If the amplitude is small and the wave length is large of a quasi-two dimensional water wave,   then the dynamics can be approximated by the Kadomtsev-Petviashvili II (KPII)  equation 
\beq\label{E:KPII-intro}
\begin{array}{c}
(-4u_{x_3}+u_{x_1x_1x_1}+6uu_{x_1})_{x_1}+3u_{{x_2}{x_2}}=0,\\ u=u({x_1},{x_2},{x_3})\in\mathbb R,~~{x_1},{x_2} \in\mathbb R, {x_3}\ge 0. 
\end{array}
 \eeq Interesting features of the water wave can be reproduced by  the KPII  line solitons which have been discovered in 1970's \cite{S76}, \cite{S79}, \cite{M79}. Precisely,   a regular KPII line soliton can be constructed by 
\beq\label{E:line-tau}
\ba{c}
u(x)=u(x_1,x_2,x_3)= 2\partial^2_{x_1}\ln\tau(x),
\ea
\eeq where the $\tau$-function is given by the Wronskian determinant
\beq\label{E:line-grassmannian}
\ba{c}
\tau(x)=\textrm{Wr}(f_1,f_2,\cdots, f_{N}),\\
f_i(x)=\sum_{j=1}^{M}a_{ij}E_j(x),\ \ 1\le i\le N,\, N<M,\\
E_j(x)=e^{\theta_j}, \, \theta_j=\kappa_j x_1+\kappa_j^2 x_2+\kappa_j^3 x_3,\ \kappa_1<\cdots<\kappa_M,
\\A=(a_{ij})\in \mathrm{Gr}(N, M)_{\ge 0}\subset \mathrm{Gr}(N, M),  
\ea
\eeq  the Grassmannian  $\mathrm{Gr}(N, M)$ denotes the set of $N$-dimensional subspaces in $\RR^M$, and $\mathrm{Gr}(N, M)_{\ge 0}$ is the subset of elements whose maximal minors   are all non-negative  \cite{BP214}, \cite{K17}. Since 2000's, there has been important progress in studying properties and classification theory of these   KPII line solitons (see \cite{K17}, \cite{K18} and references therein). Throughout this report,  \eqref{E:line-tau} defined by \eqref{E:line-grassmannian}, are called {\bf $\mathbf {\textrm{Gr}(N, M)_{\ge 0}}$ KP  solitons} for simplicity.

The well-posedness problem of the KPII equation \eqref{E:KPII-intro} 
with  initial data  $u_c(x,y)$ where $u_c(x-ct,y)$ is a KP solution  has been   solved by Molinet-Saut-Tzvetkov \cite{MST11}.  Their result shows that the deviation of the KPII solution from the initial data could evolve  exponentially.  Taking 
\beq\label{E:KdV-1soliton}
\ba{c}
N=1,\ M=2,\ \kappa_1=-\kappa_2,\ A=(1, 1) 
\ea
\eeq   which are the simplest KPII $1$-line solitons produced by the KdV $1$-soliton solutions, Mizumachi establishes excellent $L^2$-{  orbital stability  and $L^2$-                                                    instability theories for $\textrm{Gr}(1, 2)_{\ge 0}$ KP solitons  \cite{M15}.  

An important alternative approach to study the stability problem of $\textrm{Gr}(N, M)_{\ge 0}$ KP  solitons is the {\bf inverse scattering theory} (IST) based on  the {\bf Lax pair}
\begin{equation}\label{E:KPII-lax-1}
\begin{split}
\left\{ 
{\begin{array}{l}
 (-\partial_{x_2}+\partial_{x_1}^2+u )\Psi(x,\lambda)=0,\\
 (-\partial_{x_3}+ \partial_{x_1}^3+\frac 32u\partial_{x_1}+\frac 34u_{x_1}+\frac 34\partial_{x_1}^{-1}u_{x_2}\corr{-}\lambda^3   )\Psi (x,\lambda)=0 
\end{array}}
\right.
\end{split}
\end{equation}  
of the KPII equation. Indeed, the IST is to establish a bijective maps between the Lax equation 
\beq\label{E:Lax-intro}
\ba{c}(-\partial_{x_2}+\partial_{x_1}^2 
+u)\Psi(x ,\lambda)=0,
\ea\eeq(defined by the KP  solution) and a Cauchy integral equation (defined by the scattering data of the Lax equation). 
Substantial and important works  on   algebraic characterization and formal IST   have been studied by Boiti-Pempenelli-Pogrebkov-Prinari \cite{BPPP01-p}, \cite{BP302}, \cite{BP211}, \cite{BP214}, Villarroel-Ablowitz \cite{VA04}.
  In particular, the most remarkable characteristic, discontinuities for the Green function and eigenfunction of the Lax equation \eqref{E:Lax-intro}
   were discovered by Boiti, Pempenelli, Pogrebkov, and Prinari (cf \cite{BPPP01-p}, \cite{BP302}, \cite{BP214}). But a rigorous IST for perturbed $\textrm{Gr}(N, M)_{\ge 0}$ KP  solitons is still open.

Under the assumption \eqref{E:KdV-1soliton}, based on a KdV theory \cite{NMPZ84}, \cite{AC91},  rigorous analysis for the direct scattering theory of perturbed $\textrm{Gr}(1, 2)_{\ge 0}$ KP solitons  has been carried out in \cite{Wu18}. To generalize the  theory to arbitrary perturbed $\textrm{Gr}(N, M)_{\ge 0}$ KP  solitons, the Sato (or $\tau$-function) approach  \cite{BP214} is not avoidable.  The goal of this report is to 
adopt the Sato approach to provide a rigorous theory of the direct problem  for general   perturbed $\textrm{Gr}(1, 2)_{\ge 0}$ solitons which consists of all  KPII $1$-line solitons with oblique directions and phase shifts. More precisely, using the convention $x=(x_1,x_2,0)$, $k=(k_1,k_2)$, $k_j\ge 0$, $\partial_x^k=\partial_{x_1}^{k_1}\partial_{x_2}^{k_2}$, $|x|=|x_1|+|x_2|$, $|k|=k_1+k_2$, our results are stated as: for
\[
 \ba{c}
 u(x )=u_0(x)+v_0(x ),\\
 u_0(x) 
= \frac{(\kappa_1-\kappa_2)^2}2\textrm{sech}^2\frac{\theta_1(x)-\theta_2(x)-\ln a}2 ,\  \ v_0(x)\in\RR,\\
  \partial_x^k v_0\in {L^1}\cap L^\infty,    |k|\le 4,\ \ { |v_0 |_{{L^1\cap L^\infty}}\ll 1 ,}
 \ea
 \] 
\[\ba{rl}
i.&\mbox{We prove the existence of the eigenfunction of the Lax equation}\\
&\mbox{\eqref{E:Lax-intro} by establishing uniform estimates of the Green function};\\
ii.&\mbox{We justify a Cauchy integral equation for the eigenfunction by}\\
&\mbox{deriving estimates of the  spectral transform}.
\ea\]

 The contents of the paper are as follows. In Section \ref{S:oblique-green}, for a Lax equation \eqref{E:Lax-intro} defined by a  perturbed $\textrm{Gr}(1, 2)_{\ge 0}$ KP soliton,  we introduce a proper boundary data and  the Green function using the Sato theory. Then we provide     algebraic and analytic characterization, including a uniform estimate,    of the Green function.

 In Section \ref{S:oblique-eigenfunction}, we prove the existence and  study   the $\overline\partial$-scattering data of the eigenfunction, define the forward scattering transform $T$, and derive    estimates for the spectral map $\mathcal CTm$ where $\mathcal C$ is the Cauchy integral operator.  Finally, in Section \ref{S:cauchy}, we justify the initial eigenfunction satisfies a singular Cauchy integral equation and show that the singular Cauchy equation reduces to a $\textrm{Gr}(1, 2)_{\ge 0}$ KP soliton if the continuous scattering data is $0$.

{\bf Acknowledgments}. We feel  indebted to A. Pogrebkov and Y. Kodama   for introducing  the  Sato theory of the  KP hierarchy.  We would like to pay  respects to the pioneer IST theory done by Boiti, Pempinelli, Pogrebkov, Prinari. This research project was  partially supported by NSC 107-2115-M-001 -002 -.

\section{The Green function} \label{S:oblique-green}
Setting $x_3=0$, $N=1$, $M=2$, $\kappa_1<\kappa_2$,  $A=(1,a) $, $a>0$ in  \eqref{E:line-tau} and \eqref{E:line-grassmannian}, we obtain
\[
\ba{c}
\tau(x)=\det \textrm{Wr}(f_1)=e^{\theta_1}+  e^{\theta_2+\ln a}=2e^{\frac{\theta_1+\theta_2+\ln a}2}\cosh{\frac{\theta_1-\theta_2-\ln a}2}  
\ea\] and   the $\textrm{Gr}(1, 2)_{\ge 0}$ KP soliton   
\beq\label{E:line-tau-oblique}
\ba{c}
 u_0(x) 
=  \frac{(\kappa_1-\kappa_2)^2}2\textrm{sech}^2\frac{\theta_1(x)-\theta_2(x)-\ln a}2.

\ea
\eeq  
For the Lax equation \eqref{E:Lax-intro}, defined by the perturbed $\textrm{Gr}(1, 2)_{\ge 0}$ KP soliton 
\beq\label{E:potential}
\ba{c}
u(x)=u_0(x)+v_0(x),
\ea
\eeq we impose 
the boundary value data
\begin{equation}\label{E:kp-line-normal}
\ba{l}
 {\Psi(x,\lambda) \to\varphi(x,\lambda) ,\ \ x  \to\infty,}
\ea
\end{equation}  where 
\beq\label{E:sato}
\ba{rl}
&\varphi(x,\lambda)\\
=&e^{\lambda x_1+\lambda^2 x_2}\frac{\tau(x-[\frac 1\lambda])}{\tau(x)}\\
\equiv &e^{\lambda x_1+\lambda^2 x_2}\frac{e^{\theta_1(x_1-\frac 1\lambda,x_2-\frac 1{2\lambda^2},\cdots )}+  ae^{\theta_2(x_1-\frac 1\lambda,x_2-\frac 1{2\lambda^2},\cdots )}}{e^{\theta_1}+ a e^{\theta_2 }} \\
=&e^{\lambda x_1+\lambda^2 x_2}\frac{(1-\frac{\kappa_1}\lambda)e^{\theta_1 }+   (1-\frac{\kappa_2}\lambda)ae^{\theta_2}}{e^{\theta_1}+ a e^{\theta_2 }} \\
\equiv&e^{\lambda x_1+\lambda^2 x_2}\chi(x,\lambda)
\ea
\eeq  is the {\bf Sato eigenfunction} and $\chi(x,\lambda)$ is the {\bf Sato normalized eigenfunction} \cite[(2.12)]{BP214}, \cite[Theorem 6.3.8., (6.3.13) ]{D91}, \cite[Proposition 2.2, (2.21)]{K17} satisfying 
\beq\label{E;spectral}
\ba{l}
\mathcal L\varphi(x,\lambda)\equiv \left(-\partial_{x_2}+\partial^2_{x_1}+u_0(x)\right)\varphi(x,\lambda)=0,\\
L\chi(x,\lambda)\equiv \left(-\partial_{x_2}+\partial^2_{x_1}+2\lambda\partial_{x_1}+u_0(x)\right)\chi(x,\lambda)=0.
\ea
\eeq   
If we renormalize the eigenfunction $\Psi(x,\lambda)=e^{\lambda x_1+\lambda^2x_2} m(x,\lambda)$, then the boundary value problem \eqref{E:kp-line-normal} turns into
\beq\label{E:renormal}
\ba{l}
  Lm(x,\lambda)=-v_0(x)m(x,\lambda), \\
 {m(x,\lambda) \to\chi(x,\lambda) ,\ \ x  \to\infty.}
\ea
\eeq Define the Green functions $\mathcal G(x,x', \lambda)$ and $G(x,x', \lambda)$
\beq\label{E:sym-0}
\ba{c}
\mathcal L\mathcal G(x,x', \lambda) =\delta(x-x') ,\ \
L  G(x,x', \lambda)=\delta(x-x'),\\ 
\mathcal G(x,x', \lambda) =e^{ \lambda (x_1-x_1')+\lambda^2 (x_2-x_2') } G(x,x', \lambda).
\ea
\eeq

In the following, we explain one approach of Boiti et al \cite{BP214} to derive the explicit formula of the Green functions. To this aim, we introduce the {\bf Sato adjoint eigenfunction}
\beq\label{E:sato-adjoint}
\ba{rl}
&\psi(x,\lambda)\\
=&e^{-(\lambda x_1+\lambda^2 x_2)}\frac{\tau(x+[\frac 1\lambda])}{\tau(x)}\\
\equiv &e^{-(\lambda x_1+\lambda^2 x_2)}\frac{e^{\theta_1(x_1+\frac 1\lambda,x_2+\frac 1{2\lambda^2},\cdots )}+  ae^{\theta_2(x_1+\frac 1\lambda,x_2+\frac 1{2\lambda^2},\cdots )}}{e^{\theta_1}+ a e^{\theta_2 }}\\
=&e^{-(\lambda x_1+\lambda^2 x_2)}\frac{\frac {e^{\theta_1 }}{(1-\frac{\kappa_1}\lambda)}+   \frac{ae^{\theta_2}}{(1-\frac{\kappa_2}\lambda)}}{e^{\theta_1}+ a e^{\theta_2 }}\\
\equiv & e^{-[\lambda x_1+\lambda^2x_2]} \xi(x,\lambda), 
\ea
\eeq      and $\xi(x,\lambda)$, the {\bf Sato normalized adjoint eigenfunction} \cite[(2.12)]{BP214}, \cite[Theorem 6.3.8., (6.3.13) ]{D91},   satisfying  
\beq\label{E;spectral-new}
\ba{l}
\mathcal L^\dagger\psi(x,\lambda)\equiv\left(\partial_{x_2}+\partial^2_{x_1}+u_0(x)\right)\psi(x,\lambda)=0,\\
L^\dagger\xi(x,\lambda)\equiv \left(\partial_{x_2}+\partial^2_{x_1}\corr{-}2\lambda\partial_{x_1}+u_0(x)\right)\xi(x,\lambda)=0.
\ea
\eeq Note that for $\forall x\in\RR^2$ fixed, $\chi(x,\cdot)$   is a rational function  normalized at $ \infty$ and with a simple pole at $ 0$; and $\xi(x,\cdot)$ is a rational function  normalized at $ \infty$  with simple poles at $\kappa_1$, $\kappa_2$, and vanishes at $0$. Let     
\beq\label{E:residue-eigenfunction}
\ba{c}

\varphi_j(x)=\varphi(x,\kappa _j)=e^{\kappa_j x_1+\kappa_j^2x_2} \chi_j(x),
\\
\psi_j(x)=\textit{res}_{\lambda=\kappa_j}\psi(x,\lambda)=e^{-[\kappa_j x_1+\kappa_j^2x_2]}\xi_j(x).

\ea
\eeq 

\begin{lemma}\label{L:reside}
\[
\ba{c}
\sum_{j=1}^2\varphi_j (x)\psi_j(x')=0.
\ea
\]
\end{lemma}
\begin{proof} Using \eqref{E:sato} and \eqref{E:sato-adjoint}, one obtains
\[
\ba{rl}
&\varphi_1 (x)\psi_1(x')\\
=& \frac{\frac {\kappa_1-\kappa_2   }{\kappa_1}ae^{\theta_1(x)+\theta_2(x)  }}{ e^{\theta_1(x) }+ae^{\theta_2 (x) }}\frac{\kappa_1 }{ e^{\theta_1(x') }+ae^{\theta_2(x')  }}= \frac{(\kappa_1-\kappa_2)ae^{\theta_1(x)+\theta_2(x)}}{( e^{\theta_1(x) }+ae^{\theta_2 (x) })( e^{\theta_1(x') }+ae^{\theta_2 (x') })},\\
&\varphi_2 (x)\psi_2(x')\\
= &
\frac{\frac {-(\kappa_1-\kappa_2)   }{\kappa_2}e^{\theta_1(x)+\theta_2(x)  }}{ e^{\theta_1(x)}+ae^{\theta_2 (x) } }\frac{\kappa_2a }{ e^{\theta_1(x') }+ae^{\theta_2(x')  }} =\frac{-(\kappa_1-\kappa_2)ae^{\theta_1(x)+\theta_2(x)}}{( e^{\theta_1(x) }+ae^{\theta_2 (x) })( e^{\theta_1(x') }+ae^{\theta_2 (x') })}.
\ea
\] 
\end{proof}
 
\begin{lemma}\label{L:green-heat} \cite[Eq. 3.1]{BP214} Let $\theta$ be the Heaviside function. Then
\beq\label{E:green-heat} 
\ba{c}
\mathcal G(x,x',\lambda)=\mathcal G_c(x,x',\lambda)+\mathcal G_d(x,x',\lambda),\\
\mathcal G_c= -\frac{\textrm {sgn}(x_2-x_2')}{2\pi}\int_\RR\theta((s^2-\lambda_I^2)(x_2-x_2')) \varphi(x,\lambda_R+is)\psi(x',\lambda_R+is) ds,\\
\mathcal G_d=-\theta(x_2'-x_2)\{\theta(\lambda_R-\kappa_1)\varphi_1(x)\psi_1(x')+\theta(\lambda_R-\kappa_2)\varphi_2(x)\psi_2(x')\}.
\ea
\eeq 

\end{lemma}

\begin{proof} We follow the proof of \cite{BP211}. Namely, the Green function $\mathcal G$ will be constructed via an orthogonality relation of $\varphi(x,\lambda)\psi(x',\lambda)$ superposed with an appropriate cutoff function. 

To establish an orthogonality relation, note
$
\varphi(x,\lambda+i\lambda')\psi(x',\lambda+i\lambda')\\=e^{ [(\lambda+i\lambda')(x_1- x'_1)+(\lambda+i\lambda')^2(x_2- x'_2)]}\chi(x,\lambda+i\lambda')\xi(x',\lambda+i\lambda')$. 
Hence applying the Fourier inversion theorem, introducing a new variable $\lambda+i\lambda'=\lambda_R+is$, using the residue theorem, and  Lemma \ref{L:reside}, one has
\[
\ba{rl}
&\frac {\delta(x_2-x_2')}{2\pi}\int  \varphi(x,\lambda+i\lambda')\psi(x',\lambda+i\lambda')d\lambda' \\
=&\frac {\delta({x_2-x'_2})}{2\pi}  \int_\RR e^{(\lambda_R+is)(x_1-x_1') }\chi(x,\lambda_R+is)\xi(x',\lambda_R+is) ds

\ea
\]
\beq\label{E:green-c}
\ba{rl}
=&\frac {\delta({x_2-x'_2})}{2\pi}  \int_\RR e^{(\lambda_R+is)(x_1-x_1') }\left( 1+ \sum_{j=1}^2\frac{\chi_j (x)\xi_j(x')}{is+\lambda_R-\kappa_j}\right)ds\\
=&\delta({x -x' })+\corr{\delta({x_2-x'_2}) \sum_{j=1}^2e^{\kappa_j(x_1-x_1') }\chi_j (x)\xi_j(x')   }\\
&\corr{\times [\theta(x_1-x_1')\theta(\lambda_R-\kappa_j)- \theta(x_1'-x_1)\theta(\kappa_j-\lambda_R)  ]} \\
=&\delta({x -x' })+  \sum_{j=1}^2\varphi_j (x)\psi_j(x') \delta({x_2-x'_2}) \theta(\lambda_R-\kappa_j).

\ea
\eeq So we derive an orthogonality relation
\beq\label{E:ortho}
\begin{array}{rl }
  \delta({x-x'})=&\delta({x_2-x'_2})  \{ \frac 1{2\pi}   {\int_{\mathbb R} } \varphi (x,\lambda_R+is)\psi (x',\lambda_R+is) ds\\
	&- {\sum_{j=1}^2 }\varphi_j(x)\psi_j(x')\theta(  \lambda_R-\kappa_j)\}.
	\end{array}\eeq

To construct $\mathcal G$, we \corr{will superpose proper cut off functions to the orthogonal relation, so that  after applying $\mathcal L$, the cut off functions turns into $\delta({x_2-x'_2}) $.} Hence the formula $\mathcal G_d$ in \eqref{E:green-heat} is verified. Furthermore, as $G(x,x',\lambda)$ is expected to be almost bounded, one need to  confine the integral of $\varphi(x,\lambda+i\lambda')\psi(x',\lambda+i\lambda')$ to be on the region where 
$e^{ -\lambda (x_1- x'_1)-\lambda ^2(x_2- x'_2) }\newline\varphi(x,\lambda+i\lambda')\psi(x',\lambda+i\lambda')$ is bounded at $\infty$. This implies the superposed cutoff function chosen is
\[
\ba{c}
\theta(x_2-x'_2)\chi_{-}(\lambda')-\theta(x_2'-x_2)\chi_{+}(\lambda')\\
=\textrm {sgn}(x_2-x_2')\theta((s^2-\lambda_I^2)(x_2-x_2')),
\ea\]
where $\chi_{\pm}$  the characteristic function for  $\{\lambda'|\,\textit{Re}(\left[\lambda+i\lambda'\right]^2-\lambda^2)\gtrless 0\}$. Therefore the formula $\mathcal G_c$ in \eqref{E:green-heat} follows.

\end{proof}

\begin{definition}\label{D:terminology} 
For $z\in \mathfrak Z=\{ 0,\, \kappa_1,\, \kappa_2  \}$,  define 
\corr{$\kappa=\frac 12\min\{|\kappa_1|,\,|\kappa_2|,\,\kappa_2-\kappa_1\}$},
\[
\ba{c}
D_{z}=\{\lambda\in\CC\, :\, |\lambda-z|< \corr{\kappa}\},\ 
   D_z^\times=\{\lambda\in\CC\, :\, 0<|\lambda-z|<\corr{\kappa}\};\\
D_{z,r}=\{\lambda\in\CC\, :\, |\lambda-z|<  r\corr{\kappa}\},\ 
   D_{z,r}^\times=\{\lambda\in\CC\, :\, 0<|\lambda-z|< r\corr{\kappa}\} ,   
 \ea
 \] and characteristic functions  $E_z(\lambda)\equiv 1$ on $D_z$, $E_z (\lambda)\equiv 0$ elsewhere.
  Moreover, define the polar coordinate for   $D_{z,r}^\times$ to be $\{(s,\alpha)|\lambda=z+se^{i\alpha},\ 0<s<r\corr{\kappa,\,0\le\alpha \le 2\pi}\}$. 
 
Through out the report, we use $C$ to denote different uniform constants.        
\end{definition}

\begin{proposition}\label{P:eigen-green}
The Green function $ G$, defined by \eqref{E:sym-0}, satisfies the analytic constraint, \corr{for $\forall x_2-x_2'\ne 0$,   $\lambda\ne \kappa_1,\,\kappa_2$,}
\beq
\ba{c}
| G(x,x',\lambda)|\le C (1+\frac1{\sqrt{|x_2-x_2'|}} ),   \label{E:eigen-green}
\ea
\eeq   
\corr{and for any Schwartz function $ f$,

\beq\label{E:ast}
\ba{c}
\lim_{|x|\to\infty}G(x,x',\lambda)\ast f(x')  \to 0, \\
  G\ast f(x,\lambda)\equiv\iint  G(x,x',\lambda)f(x' )dx' ,\ \ dx'=dx'_1dx'_2.
\ea\eeq} 

\end{proposition}

\begin{proof}  From \eqref{E:sym-0} and Lemma \ref{L:green-heat}, one needs to show uniform estimates of $G_d(x,x', \lambda)$ and $G_c(x,x', \lambda)$ which can be written as
\beq \label{E:g-d}
\ba{rl}
& G_d(x,x', \lambda)\\
 =&-e^{\lambda(x_1'-x_1)+\lambda^2(x_2'-x_2)}\sum_{j=1}^2\theta(x_2'-x_2)\theta(\lambda_R-\kappa_j) \varphi_j(x)\psi_j(x') .
\ea\eeq and
\beq\label{E:g-c}
\ba{rl}
&G_c(x,x', \lambda) \\
= &-\frac 1{2\pi}\int_\RR ds\ \mbox{sgn}(x_2-x_2')\theta((s^2-\lambda_I^2)(x_2-x_2')) \chi(x,\lambda_R+is)\\
& \times \xi(x',\lambda_R+is) e^{ [\lambda-(\lambda_R+is)] (x_1'-x_1)+[\lambda^2-(\lambda_R+is)^2](x_2'-x_2) }\\
=&-\frac{e^{i[\lambda_I(x_1'-x_1)+2\lambda_I\lambda_R(x'_2-x_2)]}}{2\pi}
\int_\RR ds\ \mbox{sgn}(x_2-x_2')\theta((s^2-\lambda_I^2)(x_2-x_2')) \\
&\times  e^{(s^2-\lambda^2_I)(x_2'-x_2)- is[  (x_1'-x_1)+2 \lambda_R  (x_2'-x_2)] }\chi(x,\lambda_R+is)\xi(x',\lambda_R+is),\ea
\eeq

$\underline{\emph{Step 1  (Estimates for $G_d$)}}:$ From Lemma \ref{L:reside},  \eqref{E:sato}, \eqref{E:sato-adjoint},  \corr{\eqref{E:residue-eigenfunction}, \eqref{E:g-d}, the dominated convergence theorem, and the Riemann-Lebesque lemma,} one has 
\beq\label{E:discrete-residue-chi}
\ba{rl}
& |G_d(x,x', \lambda) |\\
=&|-e^{\lambda(x_1'-x_1)+\lambda^2(x_2'-x_2)}\theta(x_2'-x_2)\theta(\kappa_2-\lambda_R)\theta(\lambda_R-\kappa_1) \varphi_1(x)\psi_1(x')|\\
=&| - \frac {\theta(x_2'-x_2)\theta(\kappa_2-\lambda_R)\theta(\lambda_R-\kappa_1)\times (1-\frac{\kappa_2}{\kappa_1}) ae^{ \theta_1(x)+\theta_2(x)-\theta_\lambda(x)+\theta_\lambda(x')}\kappa_1 }{\left( e^{\theta_1(x) }+ae^{\theta_2(x)}\right)\left( e^{\theta_1(x') }+ae^{\theta_2(x')}\right)}|\\


\le &C|  \frac {\theta(x_2'-x_2)\theta(\kappa_2-\lambda_R)\theta(\lambda_R-\kappa_1) e^{  \theta_2(x)-\theta_{\lambda_R}(x)}e^{\theta_{\lambda_R}(x')-\theta_1(x')} e^{\lambda_I^2(x_2-x_2')}  }{\left( 1+ae^{\theta_2(x)-\theta_1(x)}\right)\left( 1 +ae^{\theta_2(x')-\theta_1(x')}\right)}|\\
\le& C,
\ea
\eeq and 
\beq\label{E:discrete-residue-chi-asymp}
\ba{l}
   G_d(x,x', \lambda)   
\to   \left\{
{\begin{array}{lr}
0,  &\lambda\to\kappa_2^+,\\
+\theta(x_2'-x_2)\chi_2(x)\xi_2(x'), &\lambda\to\kappa_2^-,\\
-\theta(x_2'-x_2)\chi_1(x)\xi_1(x'),  &\lambda\to\kappa_1^+,\\
0, &\lambda\to\kappa_1^-,
\end{array}}
\right.\\
\corr{\kappa_j^+=\kappa_j+0^+e^{i\alpha},\ 
\kappa_j^-=\kappa_j+0^+e^{i(\pi+\alpha)},\ \ 0\le \alpha\le \pi,}\\
\corr{\lim_{|x|\to\infty}G_d(x,x',\lambda)\ast f(x')=0, \  \textit{ for $\lambda\ne \kappa_j$ fixed}},\\
\corr{\textit{and any Schwartz function $ f$}}. 

\ea
\eeq

$\underline{\emph{Step 2  (A decomposition for $G_c$)}}:$ From
 \eqref{E:sato}, \eqref{E:sato-adjoint}, \eqref{E:g-c},
\[
\ba{rl}
&  G_c (x,x', \lambda)\\
=&-\frac{e^{i[\lambda_I(x_1'-x_1)+2\lambda_I\lambda_R(x'_2-x_2)]}}{2\pi} \int_\RR ds\ { \mbox{sgn}(x_2-x_2')\theta((s^2-\lambda_I^2)(x_2-x_2')) }\\
&\times  e^{ { (s^2-\lambda^2_I)(x_2'-x_2)}- { is[  (x_1'-x_1)+2 \lambda_R  (x_2'-x_2)] }}\\
&\times(1-\frac{\kappa_1-\kappa_2}{{ is+\lambda_R-\kappa_2}}\frac 1{1+e^{-[(\kappa_1-\kappa_2)x_1+(\kappa_1^2-\kappa_2^2)x_2 -\ln a]}})\\
&\times(1+\frac{\kappa_1-\kappa_2}{{ is+\lambda_R-\kappa_1}}\frac 1{1+e^{-[(\kappa_1-\kappa_2)x'_1+(\kappa_1^2-\kappa_2^2)x'_2 -\ln a]}}).
\ea  
\] 
So if $\lambda\in D_{\kappa_1}^c\cap D_{\kappa_2}^c$, 
\[
\ba{rl}
 |G_c (x,x', \lambda)| 
\le &C\theta(x_2-x_2')(\int^{-|\lambda_I|}_{-\infty}+\int^\infty_{|\lambda_I|})   e^{   (s^2-\lambda^2_I)(x_2'-x_2)}ds\\
&+C\theta(x'_2-x_2)\int_{-|\lambda_I|}^ {|\lambda_I| }e^{   (s^2-\lambda^2_I)(x_2'-x_2)}ds,
\ea
\] and one can either look for estimates for special functions $\mbox{erfc}(z)$ and Dawson's integral or a direct estimate  (see $\underline{\emph{Step 1}}$ in \cite{Wu18}) to derive
\beq\label{E:gauss-green}
\ba{l}
| G_c(x,x',\lambda)|\le C (1+\frac 1{\sqrt{|x_2-x_2'|}} ),  \ \lambda\in D_{\kappa_1 }^c\cap D_{\kappa_2 }^c.
\ea
\eeq \corr{Furthermore, the dominated convergence theorem and Riemann-Lebesque lemma imply
\beq\label{E:asym-g-c}
\ba{l}
 \lim_{|x|\to\infty}G_c(x,x',\lambda)\ast f(x')=0, \  \textit{ for $\lambda\ne \kappa_j$ fixed},\\
 \textit{and any Schwartz function $ f$}. 
\ea
\eeq}Hence it remains to show the estimates for $\lambda\in D_{ \kappa_j}^\times$.  

  For $\lambda\in D_{\kappa_j}^\times$, decompose
\beq\label{E:green-kappa-0}
\ba{rl}
&G_c(x,x', \lambda)= -\frac{e^{i[\lambda_I(x_1'-x_1)+2\lambda_I\lambda_R(x'_2-x_2)]}}{2\pi}\left(I_j+II_j+III_j+IV_j   \right),\ea\eeq
with
\beq\label{E:green-kappa}
\ba{rl}
I_j=&\int_{-\corr{\kappa}}^\corr{\kappa}\mbox{sgn}(x_2-x_2')\theta((s^2-\lambda_I^2)(x_2-x_2'))\chi(x, \lambda_R+is) \xi(x',\lambda_R+is)\\
&\times  [e^{is[x_1-x_1'+2\lambda_R(x_2-x_2')]+(\lambda_I^2-s^2 )(x_2-x_2')}-1]  ds,\\
II_j=&\int_{-\corr{\kappa}}^\corr{\kappa}\mbox{sgn}(x_2-x_2')\theta((s^2-\lambda_I^2)(x_2-x_2')) \\
&\times[ \chi(x, \lambda_R+is) \xi(x',\lambda_R+is)
- \frac{\chi_j(x)\xi_j(x')} {\lambda_R+is-\kappa_j}
]ds,  \\
III_j=&\int_{-\corr{\kappa}}^\corr{\kappa}\mbox{sgn}(x_2-x_2')\theta((s^2-\lambda_I^2)(x_2-x_2'))\frac{\chi_j(x)\xi_j(x')} {\lambda_R+is-\kappa_j}ds, \\
IV_j=& (\int_{-\infty}^{-\corr{\kappa}}+\int_{\corr{\kappa}}^\infty)\mbox{sgn}(x_2-x_2')\theta((s^2-\lambda_I^2)(x_2-x_2')) \chi(x, \lambda_R+is)\\
&\times \xi(x',\lambda_R+is)e^{(s^2-\lambda^2_I)(x_2'-x_2)- is[  (x_1'-x_1)+2 \lambda_R  (x_2'-x_2)] }  ds. 
\ea
\eeq  

By the same method as that for \eqref{E:gauss-green}, we have
\beq\label{E:green-35}
\ba{c}
|II_j|<C,\ |IV_j| \le C (1+\frac1{\sqrt{|x_2-x_2'|}} ),\  \,\lambda\in D_{\kappa_j}.
\ea
\eeq 

$\underline{\emph{Step 3  (Estimates for $III_j$)}}:$ 
Since
\[
\ba{rl}
 &III_j \\
=&\theta(x_2-x_2')(\int_{-\corr{\kappa}}^{-|\lambda_I|}+\int_{|\lambda_I|}^{\corr{\kappa}})\frac{\chi_j(x) \xi_j(x')}{\lambda_R+is-\kappa_j}ds-\theta(x_2'-x_2)\int_{-|\lambda_I|}^{ |\lambda_I|}\frac{\chi_j(x) \xi_j(x')}{\lambda_R+is-\kappa_j}ds\\
=& -i\chi_j(x) \xi_j(x') \{\theta(x_2-x_2')  \int_{-\corr{\kappa}}^\corr{\kappa}  \frac{1}{s-i(\lambda_R-\kappa_j)}ds   -  \int_{-|\lambda_I|}^{|\lambda_I|}\frac{1}{s-i(\lambda_R-\kappa_j)}ds\}.
\ea
\]
 By logarithmic function, 
\beq\label{E:approach} 
\ba{rl}
& \lim_{\lambda\to\kappa _j}\int_{-\corr{\kappa}}^{\corr{\kappa}}\frac{1}{s-i(\lambda_R-\kappa_j)}ds=i\pi(2\theta( \lambda_R-\kappa_j)-1), \\ 
&\int_{-|\lambda_I|}^{|\lambda_I|}\frac{1}{s-i(\lambda_R-\kappa_j)}ds
 
= 
2 \pi i[\theta(\lambda_R-\kappa_j)-1] +2i\cot^{-1}\frac{\lambda_R- \kappa_j}{ |\lambda_I|}, \ \lambda\in D_{\kappa_j}^\times . 
\ea
\eeq  Here
\beq
\ba{l}
\cot^{-1}\frac {\lambda_R-\kappa_j}{|\lambda_{I}|}
= \left\{
{\begin{array}{lcl}
   \alpha,&0<\alpha\le\pi , & \lambda\in D_{\kappa_j}^\times,\\
2\pi-\alpha,&\pi\le\alpha<2\pi,& \lambda\in D_{\kappa_j}^\times,
\end{array}}
\right.  \ea \label{E:arccot}\eeq
As a result,
\beq\label{E:green-4-bdd}
\ba{c}
|III_j|\le C,\quad\lambda\in D_{\kappa_j}^\times,
\ea
\eeq and 
\beq\label{E:green-4-jump}
\ba{rl}
&-\frac{e^{i[\lambda_I(x_1'-x_1)+2\lambda_I\lambda_R(x'_2-x_2)]}}{2\pi} III _j\\
\to & [-1+\frac 12\theta(x_2-x_2')]\chi_j(x) \xi_j(x') +\theta(\lambda_R-\kappa_j) \theta(x_2'-x_2 )\chi_j(x) \xi_j(x') \\
&+\frac 1\pi \chi_j(x) \xi_j(x')\cot^{-1}\frac{\lambda_R- \kappa_j}{ |\lambda_I|},\ \ \textit{as $\lambda\to\kappa_j$.} 
\ea\eeq

$\underline{\emph{Step 4  (Estimates for $I_j$)}}:$ We follow the same method as that  in \cite{Wu18} to derive estimates for $I_j$. Setting $y_1=x_1-x_1'$, $y_2=x_2-x_2'$,   estimates for $I_j$ are reduced to 
\beq\label{E:out-in}
\ba{rl}
I_j^{in}(y_1,y_2,\lambda)=&-\theta(-y_2)\int_{-|\lambda_I|}^{|\lambda_I|}
 \frac{e^{(\lambda_I^2-s^2 )y_2+is(y_1+2\lambda_Ry_2)}-1}{s-i(\lambda_R-\kappa_j)}ds,\\
I_j^{out}(y_1,y_2,\lambda)=& \theta( y_2)(\int_{-\corr{\kappa}}^{-|\lambda_I|}+\int_{|\lambda_I|}^{\corr{\kappa}})
 \frac{e^{(\lambda_I^2-s^2 )y_2+is(y_1+2\lambda_Ry_2)}-1}{s-i(\lambda_R-\kappa_j)}ds

\ea
\eeq In this step, we study $  I_j^{in}$ by considering  cases
\beq\label{E:I-cases}
\ba{rl}
(1a)&(\lambda_R- \kappa_j)(y_1 +2\lambda_Ry_2)\ge 0,\  
 { \left|\, |\lambda_I|-|\lambda_R-\kappa_j  |\,\right|}\le \frac 12|\lambda_I|,\\
(1b)& (\lambda_R- \kappa_j)(y_1 +2\lambda_Ry_2)\ge 0, \  
 { \left|\, |\lambda_I|-|\lambda_R- \kappa_j  |\,\right|}\ge \frac 12|\lambda_I| ,\\
(1c)&(\lambda_R- \kappa_j)(y_1 +2\lambda_Ry_2)< 0.

\ea
\eeq 
 In Case (1a),  {$|\lambda_R-\kappa_j|\ge \frac {|\lambda_I|}2$}. So
\beq\label{E:I-in-a}
\ba{rl}
 | I_j^{in}|
\le &  \theta(-y_2)\int_{-|\lambda_I|}^{|\lambda_I|}| \frac{e^{(\lambda_I^2-s^2 )y_2+is (y_1 +2\lambda_Ry_2)}   -1}{ i(\lambda_R-\kappa_j)} | ds\\ \le & C  \int_{-|\lambda_I|}^{|\lambda_I|} \frac{1}{|\lambda_I|}  ds \le C.
\ea\eeq 

In Case (1b) or (1c), we deform the real interval $-|\lambda_I|\le s\le |\lambda_I|$ to the semicircle   $\Gamma$, defined by
\[
\ba{l} \Gamma\equiv\{\mathfrak s=se^{i\beta}\in\CC: s=|\lambda_I|,\,(y_1 +2\lambda_Ry_2)\sin\beta >0\}\subset \Omega, \\
\Omega\equiv \{\mathfrak s=se^{i\beta}\in\CC: 0\le s\le|\lambda_I|,\,(y_1 +2\lambda_Ry_2)\sin\beta >0\},
\ea\] and note that
\[
\ba{c}
\textit{$e^{(\lambda_I^2-s^2 )y_2+is (y_1  +2\lambda_Ry_2)}$ is uniformly bounded on the  half disk  $\Omega$},\\
\textit{$i(\lambda_R-\kappa_j)\in\Omega$ in Case (1b)}.
\ea\] 
Besides,   (1b) or (1c) implies
\[
\ba{c} {\mbox{Distance} \{\Gamma, i(\lambda_R-\kappa_j)\}\ge\frac {|\lambda_I|}2.}
\ea\]
Therefore, 
\beq\label{E:I-in-bc}
\ba{rl}
 | I_j^{in}|
= &|\theta(-y_2)\int_{-|\lambda_I|}^{|\lambda_I|} \frac{e^{(\lambda_I^2-s^2 )y_2+is (y_1  +2\lambda_Ry_2)}   -1}{s-i(\lambda_R-\kappa_j)}  ds|\\
\le &|\theta(-y_2)\int_{\Gamma} \frac{e^{(\lambda_I^2-s^2 )y_2+is (y_1  +2\lambda_Ry_2)}   -1}{s-i(\lambda_R-\kappa_j)}  ds|\\
&+C\delta((\lambda_R- \kappa_j)(y_1 +2\lambda_Ry_2))\theta(-y_2)\delta(|\lambda_I|-|\lambda_R- \kappa_j  |)\\
&\times  e^{(\lambda_I^2+(\lambda_R-\kappa_j)^2 )y_2-(\lambda_R-\kappa_j) (y_1  +2\lambda_Ry_2)}+C\\
\le &2| \int_{\Gamma} \frac{C}   {|\lambda_I|}  ds|+C\le C.
\ea\eeq
Hence estimates for $I_j^{in}$ follows from \eqref{E:I-in-a} and \eqref{E:I-in-bc}.

$\underline{\emph{Step 5  (Estimates for $I_j$ (continued))}}:$ We consider the following cases,
\[
\ba{rl}
(2a)&|\lambda_I|\sqrt{y_2}\ge 1,\\
(2b)&|\lambda_I|\sqrt{y_2}\le 1 . 
\ea
\]

  In case of (2a),  let $\xi=\frac s{|\lambda_I|}$,   
\beq\label{E:I-out-a}
\ba{rl}
 |   {I}_{j}^{out} |   
=  &C |   \theta(y_2)\{(\int_{-\frac {\corr{\kappa}}{|\lambda_I|}}^{-1}+\int_{1}^{\frac {\corr{\kappa}}{|\lambda_I|}})  \frac { e^{{\lambda_I^2y_2(1-\xi^2 )+i|\lambda_I|\xi(y_1+2\lambda_Ry_2)}}   }{ \xi -i\frac{\lambda_R- {\kappa_j}}{|\lambda_I|}} d\xi \\
 & -(\int_{-\corr{\kappa}}^{-|\lambda_I|}+\int_{|\lambda_I|}^{\corr{\kappa}})  \frac {  1}{ s-i(\lambda_R-  {\kappa}_j )} ds\}|\\ 
\le &C(\int_{-\frac {\corr{\kappa}}{|\lambda_I|}}^{-1}+\int_{1}^{\frac {\corr{\kappa}}{|\lambda_I|}})   e^{  1-\xi^2 }     d\xi+C\\
=&Ce\sqrt {e^{-1}-e^{-\corr{\frac { \kappa^2}{|\lambda_I|^2}}}}+C\le C.
\ea
\eeq 

  In case (2b),  let $s=\frac \omega{\sqrt {y_2}}$,
\beq\label{E:I-out-b}
\ba{rl}
  | {I}_{j}^{out} |  
= & C|   \theta(y_2)(\int_{-1}^{-|\lambda_I|\sqrt{y_2}}+\int_{|\lambda_I|\sqrt{y_2}}^{1}) \frac {e^{i\omega\frac {y_1+2\lambda_Ry_2}{\sqrt{y_2}}+\lambda_I^2y_2-\omega^2 } }{\omega-{ i(\lambda_R- {\kappa}_j)\sqrt{y_2}}}  d\omega  \\
&+ \theta(y_2)(\int_{-\corr{\kappa}\sqrt{y_2}}^{-1}+\int_{1}^{\corr{\kappa}\sqrt{y_2}})\frac {e^{i\omega\frac {y_1+2\lambda_Ry_2}{\sqrt{y_2}}+\lambda_I^2y_2-\omega^2 }}{\omega-{ i(\lambda_R- {\kappa}_j)\sqrt{y_2}}}  d\omega  \\
&\corr{-}  {  
 (\int_{-\corr{\kappa}}^{-|\lambda_I|}+\int_{|\lambda_I|}^{\corr{\kappa}})\frac {1}{s-i(\lambda_R-{\kappa}_j)} ds} | 
\\
\le &C\corr{(A_1+A_2+A_3)},
 \ea\eeq
where
\beq\label{E:I-out-b-A}
\ba{rl}
A_1=&|   \theta(y_2)\corr{(\int_{-1}^{-|\lambda_I|\sqrt{y_2}}+\int_{|\lambda_I|\sqrt{y_2}}^{1}) }\frac {e^{i\omega\frac {y_1+2\lambda_Ry_2}{\sqrt{y_2}}+\lambda_I^2y_2-\omega^2 }-e^{i\omega\frac {y_1+2\lambda_Ry_2}{\sqrt{y_2}} }}{\omega-{ i(\lambda_R- {\kappa}_j)\sqrt{y_2}}}  d\omega\\
&\corr{+(\int_{-1}^{-|\lambda_I|\sqrt{y_2}}+\int_{|\lambda_I|\sqrt{y_2}}^{1})  \frac { e^{i\omega\frac {y_1+2\lambda_Ry_2}{\sqrt{y_2}} }}{\omega-{ i(\lambda_R- {\kappa}_j)\sqrt{y_2}}}  d\omega) }|\\
\le &A_{11}+A_{12},\\
A_2=&|   \theta(y_2)(\int^{-1}_{-\corr{\kappa}\sqrt{y_2}}  +\int^{\corr{\kappa}\sqrt{y_2}}_{1} )\frac { \corr{ e^{i\omega\frac {y_1+2\lambda_Ry_2}{\sqrt{y_2}}+\lambda_I^2y_2-\omega^2 }}}{\omega-{ i(\lambda_R- {\kappa}_j)\sqrt{y_2}}}  d\omega) |,\\
\corr{A_3}=&|  (\int^{-|\lambda_I|}_{-\corr{\kappa}}+\int^ {\corr{\kappa}}_{|\lambda_I|})\frac {1}{s-i(\lambda_R-{\kappa}_j)} ds  |\le C.
\ea
\eeq \corr{ Since $|\lambda_I|\sqrt{y_2}\le 1$,
\[
\ba{rl}
 A_2\le &C (\int^{-1}_{-\kappa\sqrt{y_2}}  +\int^{\kappa\sqrt{y_2}}_{1} )  |  e^{i\omega\frac {y_1+2\lambda_Ry_2}{\sqrt{y_2}}+\lambda_I^2y_2-\omega^2 } |  d\omega 
\\
= &C (\int^{-1}_{-\kappa\sqrt{y_2}}  +\int^{\kappa\sqrt{y_2}}_{1} )    e^{ \lambda_I^2y_2-\omega^2 }   d\omega \\
\le &C (\int^{-1}_{-\kappa\sqrt{y_2}}  +\int^{\kappa\sqrt{y_2}}_{1} )    e^{  -\omega^2 }   d\omega \\
\le &C.
\ea\]}Applying the mean value theorem to $A_{11}$, 
\[
\ba{rl}
 A_{11}\equiv &|   \theta(y_2)\corr{(\int_{-1}^{-|\lambda_I|\sqrt{y_2}}+\int_{|\lambda_I|\sqrt{y_2}}^{1}) }\frac {e^{i\omega\frac {y_1+2\lambda_Ry_2}{\sqrt{y_2}}+\lambda_I^2y_2-\omega^2 }-e^{i\omega\frac {y_1+2\lambda_Ry_2}{\sqrt{y_2}} }}{\omega-{ i(\lambda_R- {\kappa}_j)\sqrt{y_2}}}  d\omega|
 \\ 
\le & \theta(y) (\int_{|\lambda_I|\sqrt{y}}^{1}   |\frac {e^{i\omega\frac { x+2\lambda_Ry}{\sqrt{y}}+\lambda_I^2y-\omega^2 }-e^{i\omega\frac { x+2\lambda_Ry}{\sqrt{y}}}}{ {\omega-|\lambda_I|\sqrt{y}}}  | d\omega \\
 &+ \int^{-|\lambda_I|\sqrt{y}}_{-1}   |\frac {e^{i\omega\frac { x+2\lambda_Ry}{\sqrt{y}}+\lambda_I^2y-\omega^2 }-e^{i\omega\frac { x+2\lambda_Ry}{\sqrt{y}}}}{ {\omega+|\lambda_I|\sqrt{y}}}  | d\omega )\\
 \le & \theta(y) (\int_{|\lambda_I|\sqrt{y}}^{1}   |\frac {e^{ \lambda_I^2y-\omega^2 }-1}{ {\omega-|\lambda_I|\sqrt{y}}}  | d\omega  +  \int^{-|\lambda_I|\sqrt{y}}_{-1}   |\frac {e^{ \lambda_I^2y-\omega^2 }-1}{ {\omega+|\lambda_I|\sqrt{y}}}  | d\omega) \\
 \le &C,
\ea
\]
\corr{and   the argument of (1a), (1b), (1c) in $\underline{\emph{Step 4}}$   to 
\corr{\[
\ba{c}
A_{12}=|(\int_{-1}^{-|\lambda_I|\sqrt{y_2}}+\int_{|\lambda_I|\sqrt{y_2}}^{1})  \frac { e^{i\omega\frac {y_1+2\lambda_Ry_2}{\sqrt{y_2}} }}{\omega-{ i(\lambda_R- {\kappa}_j)\sqrt{y_2}}}  d\omega   |,
\ea\]}one can derive uniform boundedness for $A_{12}$. 
Therefore, we have justify $|I_j^{out}|\le C$ in case (2b). }

Combining $\underline{\emph{Step 1}}$ to $\underline{\emph{Step 5}}$, we prove \eqref{E:eigen-green}.

\end{proof}

\begin{lemma}\label{L:discontinuities}
The Green function $ G$, defined by \eqref{E:sym-0}, satisfies the algebraic constraint
\beq\label{E:green-sym}
\ba{c}
G(x,x',\lambda)=\overline{G(x,x',\overline\lambda)}.
\ea
\eeq  Moreover, there exist $ {\mathfrak G}_j $, $\omega_j$, such that
\beq\label{E:g-asym-pm-i-prep-limit-sym}
\begin{array}{rl}
  &G(x,x', \lambda)\\
=  & \left\{
{\begin{array}{ll}
{\mathfrak G}_1(x,x')+\frac 1\pi 
\chi_1(x)\xi_1(x')\cot^{-1}\frac{\lambda_R-\kappa_1}{|\lambda_I|}+\omega_1(x,x', \lambda) , &\lambda\in D_{ \kappa_1}^\times, \\
{\mathfrak G}_2(x,x')-\frac 1\pi \chi_2(x)\xi_2(x')\cot^{-1}\frac{\kappa_2-\lambda_R}{|\lambda_I|} +\omega_2(x,x', \lambda) , &\lambda\in D_{ \kappa_2}^\times,
\end{array}}
\right.   
   
\end{array}
\eeq  with $\cot^{-1}\frac{\lambda_R-\kappa_1}{|\lambda_I|}$, $\cot^{-1}\frac{\kappa_2-\lambda_R}{|\lambda_I|}$ defined by \eqref{E:arccot},
\begin{gather}
\ba{c}
|{\mathfrak G}_j  |_{L^\infty(D_{ {\kappa}_j})},\  |\omega_j  |_{L^\infty(D_{ {\kappa}_j})}\le C (1+\frac1{\sqrt{|x_2-x_2'|}} ),\\ 
 |\frac {\omega_j (x,x',\lambda)}{\lambda-\kappa_j} |_{L^\infty(D_{ {\kappa}_j})}\le C(1+\frac1{\sqrt{|x_2-x_2'|}} )(1+|x' -x |) ,
\ea \label{E:green-variation-derivative} 
\end{gather} 
and  the symmetry
\begin{equation}\label{E:green-exp-1-g}
\ba{c}
 {\mathfrak G}_2(x,x')e^{ \kappa_2(x_1-x'_1)+  \kappa_2^2(x_2-x'_2)}=e^{ \kappa_1( x _1-x _1')+ \kappa_1^2(x _2-x_2')}{\mathfrak G}_1(x,x'),\\
  G(x,x' , \kappa_2+0^+e^{i\alpha})e^{ \kappa_2(x_1-x'_1)+  \kappa_2^2(x_2-x'_2)}\\=e^{ \kappa_1( x _1-x _1')+ \kappa_1^2(x _2-x_2')} G(x,x' ,\kappa_1+0^+e^{i(\pi+\alpha)}) .  
\ea
\end{equation}  
\end{lemma}

\begin{proof} $\underline{\emph{Step 1 (Proof for \eqref{E:green-sym})}}:$ First of all, applying \eqref{E:sato}, \eqref{E:sato-adjoint}, \eqref{E:residue-eigenfunction}, Lemma \ref{L:green-heat}, and by a change of variables $s\mapsto -s$, one can prove the algebraic constraint \eqref{E:green-sym}.

$\underline{\emph{Step 2 (Proof for \eqref{E:g-asym-pm-i-prep-limit-sym})}}:$ For fixed $x$, $x'$, asymptotic \eqref{E:g-asym-pm-i-prep-limit-sym} can be obtained via the 
dominated convergence theorem, \eqref{E:green-kappa}, \eqref{E:discrete-residue-chi-asymp},  \eqref{E:green-4-jump}, estimates of (2b) in $\underline{\emph{Step 5}} $ of Proposition \ref{P:eigen-green}, and definition \eqref{E:arccot}. Moreover,   the error estimate \eqref{E:green-variation-derivative} follows from a similar argument (more elaborating) as that for deriving \eqref{E:eigen-green}. We omit the details for simplicity and refer \cite[Lemma 3.1]{Wu18} for a similar detailed proof.

$\underline{\emph{Step 3  (Proof of \eqref{E:green-exp-1-g})}}:$ From \eqref{E:arccot}, it suffices to establish 
\beq\label{E:G-sym}
\ba{c}\mathcal G(x,x' , \kappa_2+0^+e^{i\alpha})=\mathcal G(x,x' ,\kappa_1+0^+e^{i(\pi+\alpha)}) .
\ea
\eeq We now exploit the approach in \cite[Proposition 9 (i)] {VA04} to prove \eqref{E:G-sym}. For fixed  $x \ne 0$, $0<\alpha\le 2\pi$,   let 
\beq\label{E:G-sym-d}
\ba{l}
\lambda_2=\lambda_{2,R}+i\lambda_{2,I}=\kappa_2+0^+e^{i\alpha},\\
\lambda_1=\lambda_{1,R}+i\lambda_{1,I}=\kappa_1+0^+e^{i(\pi+\alpha)}. \ea
\eeq Then from Lemma \ref{L:reside} and Lemma \ref{L:green-heat}, immediately, one has
\beq\label{E:G-sym-d-0} 
\ba{rl}
(i) &\textit{for } 0<\alpha<\frac \pi 2 \textit{ or }\, \frac{3\pi}2<\alpha<2\pi,\\
&\mathcal G_d(x,x' , \lambda_2)=\mathcal G_d(x,x' ,\lambda_1)=0; \\ 
(ii) &\textit{for   $\frac \pi 2<\alpha<\frac{3\pi}2$},\\
&\mathcal G_d(x,x' ,\lambda_2)=\mathcal G_d(x,x' ,\lambda_1)= -\theta(x_2'-x_2)\varphi_1(x)\psi_1(x') .
\ea
\eeq On the other hand,
\beq\label{E:G-sym-+}
\ba{rl}
\corr{(-2\pi)} \mathcal G_c(x,x',\lambda_2)    
= &\theta(x_2-x_2') \int_{\mathbb R} \varphi (x,is+ \lambda_{2,R})\psi (x',is+ \lambda_{2,R})ds\\
& -  \int_{-|\lambda_{2,I}|}^{|\lambda_{2,I}|} \varphi (x,is+ \lambda_{2,R})\psi (x',is+ \lambda_{2,R})ds;\\
\corr{(-2\pi)} \mathcal G_c(x,x',\lambda_1)    
= &\theta(x_2-x_2') \int_{\mathbb R}\varphi (x,is+ \lambda_{1,R})\psi (x',is+ \lambda_{1,R})ds \\
& -  \int_{-|\lambda_{1,I}|}^{|\lambda_{1,I }|} \varphi (x,is+ \lambda_{1,R})\psi (x',is+ \lambda_{1,R})ds.
\ea
\eeq Deforming the contour, applying the residue theorem and Lemma \ref{L:reside}, 
\beq\label{E:G-sym-+-1}
\ba{rl}
&\theta(x_2-x_2') \int_{\mathbb R} \varphi (x,is+ \lambda_{2,R})\psi (x',is+ \lambda_{2,R})ds\\
=&\theta(x_2-x_2') \int_{\RR} \varphi (x,is+\frac{\kappa_1 +\kappa_2}2)\psi (x',is+\frac{\kappa_1 +\kappa_2}2)ds\\
&+(-2\pi)\theta(x_2-x_2') [1- \theta(   \kappa_2-\lambda_{2,R} ) ] \varphi_2(x )\psi_2(x' )\\
=&\theta(x_2-x_2') \int_{\RR} \varphi (x,is+\frac{\kappa_1 +\kappa_2}2 )\psi (x',is+\frac{\kappa_1 +\kappa_2}2 )ds\\
&-(-2\pi)\theta(x_2-x_2') [1-  \theta( \lambda_{1,R}-\kappa_1) ] \varphi_1(x )\psi_1(x' )\\
=&\theta(x_2-x_2') \int_{\mathbb R} \varphi (x,is+ \lambda_{1,R})\psi (x',is+ \lambda_{1,R})ds.
\ea
\eeq
 On the other hand, the residue theorem,  Lemma \ref{L:reside}, \eqref{E:arccot},    \eqref{E:approach}, \eqref{E:G-sym-d}, and the dominated convergence theorem imply
\beq\label{E:G-sym-+-2}
\ba{rl}
&\int_{-|\lambda_{2,I}|}^{|\lambda_{2,I}|} \varphi (x,is+ \lambda_{2,R})\psi (x',is+ \lambda_{2,R})ds\\
=&+\corr{(- i )}\varphi_2(x )\psi_2(x' )\int_{-|\lambda_{2,I}|}^{|\lambda_{2,I}|}  \frac{1}{s-i( \lambda_{2,R}-\kappa_2)}d\eta\\
=&+\corr{(- i )}\varphi_1(x )\psi_1(x' )\int_{-|\lambda_{1,I}|}^{|\lambda_{1,I }|} \frac{1}{s-i( \lambda_{1,R}-\kappa_1)}d\eta\\
=&\int_{-|\lambda_{1,I}|}^{|\lambda_{1,I }|} \varphi (x,is+ \lambda_{1,R})\psi (x',is+ \lambda_{1,R})ds.
\ea
\eeq
Consequently,  \eqref{E:G-sym}  follows  from \eqref{E:G-sym-d}-\eqref{E:G-sym-+-2}.
\end{proof}

\begin{lemma} \label{L:cont-debar} \cite{BP214}
For $\lambda_I\ne 0$,
\[
\ba{c}
 \partial_{\bar\lambda}  G(x,x' ,\lambda)=
  
 -\frac {\mbox{sgn}(\lambda_I)}{2\pi i}e^{(\overline\lambda-\lambda)(x_1-x_1')+(\overline\lambda ^2-\lambda^2)(x_2-x_2')} \chi(x,\overline\lambda )\xi(x',\overline\lambda).
\ea
\] 

\end{lemma}

\begin{proof} \corr{ Using the change of variables $\lambda+i\lambda'=\lambda_R+is$ in \eqref{E:green-heat}},  
\[
\ba{rl}
&\partial_{\overline\lambda}  G _c(x,x', \lambda)\\
=&-\frac {\theta(x_2-x'_2)}{2\pi}\\
&\times\left\{
{\begin{array}{l }
\partial_{\overline\lambda}\{[\int_{-\infty}^{-2\lambda_I }+\int_{0}^\infty]\chi(x,\lambda+i\lambda')\xi(x',\lambda+i\lambda')   \\
 \times  e^{[(\lambda+i\lambda')-\lambda](x_1-x_1')+[(\lambda+i\lambda')^2-\lambda^2](x_2-x_2')} d\lambda'\}\\
 +\frac {\theta(-(x_2-x'_2))}{2\pi}\partial_{\overline\lambda}\{ \int_{-2\lambda_I }^{0}\chi(x,\lambda+i\lambda')\xi(x',\lambda+i\lambda')\\
 \times  e^{[(\lambda+i\lambda')-\lambda](x_1-x_1')+[(\lambda+i\lambda')^2-\lambda^2](x_2-x_2')}d\lambda'\} ,\quad\quad
\mbox{if $\lambda_I>0$, }\\
\partial_{\overline\lambda}\{[\int_{-\infty}^0+\int_{-2\lambda_I }^\infty] \chi(x,\lambda+i\lambda')\xi(x',\lambda+i\lambda')\\
 \times  e^{[(\lambda+i\lambda')-\lambda](x_1-x_1')+[(\lambda+i\lambda')^2-\lambda^2](x_2-x_2')} d\lambda'\}\\
 +\frac {\theta(-(x_2-x'_2))}{2\pi}\partial_{\overline\lambda}\{ \int^{-2\lambda_I }_{0}\chi(x,\lambda+i\lambda')\xi(x',\lambda+i\lambda')\\
 \times  e^{[(\lambda+i\lambda')-\lambda](x_1-x_1')+[(\lambda+i\lambda')^2-\lambda^2](x_2-x_2')}d\lambda'\}
,\quad \quad\mbox{if $\lambda_I<0$  }  
\end{array}}
\right.
\ea\]
\[\ba{rl}
=&-\frac {\mbox{sgn}(\lambda_I)}{2\pi i}\delta_{\lambda'=-2\lambda_I}e^{[(\lambda+i\lambda')-\lambda](x_1-x_1')+[(\lambda+i\lambda')^2-\lambda^2](x_2-x_2')} \\
&\times\chi(x,\lambda+i\lambda')\xi(x',\lambda+i\lambda')\\
&-\frac 1{2\pi}\int_\RR e^{[(\lambda+i\lambda')-\lambda](x_1-x_1')+[(\lambda+i\lambda')^2-\lambda^2](x_2-x_2')}\Xi(x_2-x'_2,\lambda,\lambda_I+\lambda')\\ 
&\times\partial_{\overline\lambda}[\chi(x,\lambda+i\lambda')\xi(x',\lambda+i\lambda')]d\lambda' \ea\]
\corr{Here
\beq\label{E:f}
\ba{rl}
&\Xi(x_2,\lambda,s)=\theta(x_2)\chi_{\{-\infty,-|\lambda_I|\}\cup \{|\lambda_I|,\infty\}}(s)-\theta(-x_2)\chi_{\{-|\lambda_I|,|\lambda_I|\}}(s),\\
&\textit{$\chi_A(s)$ denotes the characteristic function of the set $A$.}
\end{array}
\eeq} Using
$
\frac 1{\pi}\partial_{\bar \lambda}\left(\frac 1{\lambda-a}\right)=\delta_{\lambda_R=a_R}\delta_{\lambda_I=a_I}$,  one has
\[
\ba{rl}
 &\partial_{\overline\lambda}[\chi(x,\lambda+i\lambda')\xi(x',\lambda+i\lambda')] 
\\= &
  \corr{\pi \chi_1(x)\xi_1(x')\delta_{\lambda_R=\kappa_1}\delta_{\lambda'=-\lambda_I}  
 + \pi \chi_2(x)\xi_2(x')\delta_{\lambda_R=\kappa_2}\delta_{\lambda'=-\lambda_I}.}
\ea
\]So 
\beq\label{E:dbar-g-c}
\ba{rl}
&\partial_{\overline\lambda}  G _c(x,x', \lambda)\\
=&-\frac {\mbox{sgn}(\lambda_I)}{2\pi i}e^{[\overline\lambda-\lambda](x_1-x_1')+[\overline\lambda ^2-\lambda^2](x_2-x_2')} \chi(x,\overline\lambda )\xi(x',\overline\lambda)\\
&\corr{+\frac 1{2 } \theta(x'_2-x_2)  
   e^{-i\lambda_I(x_1-x_1')+(-2i\kappa_1\lambda_I+\lambda_I^2)(x_2-x_2')}\chi_1(x)\xi_1(x')\delta_{\lambda_R=\kappa_1}} \\
&\corr{+ \frac 1{2 } \theta(x'_2-x_2) e^{-i\lambda_I(x_1-x_1')+(-2i\kappa_2\lambda_I+\lambda_I^2)(x_2-x_2')}\chi_2(x)\xi_2(x')\delta_{\lambda_R=\kappa_2}} \\
=&-\frac {\mbox{sgn}(\lambda_I)}{2\pi i}e^{[\overline\lambda-\lambda](x_1-x_1')+[\overline\lambda ^2-\lambda^2](x_2-x_2')} \chi(x,\overline\lambda )\xi(x',\overline\lambda)\\
&+\frac 1{2 } \theta(x'_2-x_2)   e^{-\lambda(x_1-x_1') -\lambda ^2 (x_2-x_2')} \varphi_1(x)\psi_1(x')\delta_{\lambda_R=\kappa_1}\\
&+\frac 1{2 } \theta(x'_2-x_2) e^{-\lambda(x_1-x_1') -\lambda ^2 (x_2-x_2')} \varphi_2(x)\psi_2(x')\delta_{\lambda_R=\kappa_2}.
\ea
\eeq 
 
 Besides, a direct computation yields
\beq\label{E:dbar-g-d}
\ba{rl}
&\partial_{\overline\lambda}  G _d(x,x', \lambda)\\
=&-\partial_{\overline\lambda}[e^{ -\lambda (x_1-x_1')-\lambda^2 (x_2-x_2') } 
\theta(x_2'-x_2)\\
&\times\{\theta(\lambda_R-\kappa_1)\varphi_1(x)\psi_1(x')+\theta(\lambda_R-\kappa_2)\varphi_2(x)\psi_2(x')\}] \\
=&-\frac {e^{ -\lambda (x_1-x_1')-\lambda^2 (x_2-x_2') }}2 
\theta(x_2'-x_2)\varphi_1(x)\psi_1(x')\delta_{\lambda_R=\kappa_1}\\
& -\frac {e^{ -\lambda (x_1-x_1')-\lambda^2 (x_2-x_2') } }2
\theta(x_2'-x_2)
\varphi_2(x)\psi_2(x')\delta_{\lambda_R=\kappa_2}.

\ea
\eeq

 \corr{Combining \eqref{E:dbar-g-c} and \eqref{E:dbar-g-d}, we prove the lemma.}
\end{proof}

\section{The eigenfunction  and spectral transformation}\label{S:oblique-eigenfunction}
 
\corr{Based on the characterization of the Green function $G$, we can provide the $\overline\partial$ data of $m$ in Theorem \ref{T:KP-eigen-existence} and \ref{T:sd-continuous}.}
\begin{theorem} \label{T:KP-eigen-existence} If $\partial_x^kv_0 \in  {L^1\cap L^\infty}$, $   |k|\le 2 $, $|v_0|_{{L^1\cap L^\infty}}\ll 1$, $v_0(x)\in\RR$, then for fixed $\lambda\in \CC \backslash\{ 0,\kappa_1,\kappa_2\}$, there is a   unique solution  $m(x, \lambda)$ to the spectral equation
\beq\label{E:direct-spectral}
\ba{c}
Lm(x,\lambda)
=- v_0(x)m(x,\lambda),\\
\lim_{|x|\to\infty}  (   m(x,\lambda)-\chi(x,\lambda))=0, 

\ea
\eeq where the spectral operator $L$ and $\chi$ are defined by  
 \eqref{E:sato} and \eqref{E:sym-0}. 
 
Moreover, $m(x,\lambda)=\overline{m(x,\overline\lambda)}$, and for fixed $x\in\RR^2$,  $m(x,\lambda)$ satisfies
\begin{gather}
\ba{c}
{ |(1- E_{0 }  )  m(x, \lambda)|\le C}|v_0|_{L^1\cap L^\infty};
\ea\label{E:bdd}\\
\ba{c}
m(x,\lambda)=  { \frac{  m_{res}(x )}{\lambda } } +  m_{0,r}(x, \lambda),
\ \lambda\in D^ \times_{0} ,\\ 
  m_{res}(x )\in \RR ,  \ \ |  m_{res}|_{L^\infty}\le C|v_0|_{L^1\cap L^\infty},\\
|  \lambda m_{0,r} |_{L^\infty\corr{(D_0)}},\,| \frac{ m_{0,r}}{1+|x|} |_{L^\infty\corr{(D_{0})}}\le C  |v_0|_{L^1\cap L^\infty}  ;
\ea\label{E:pole}\\
\nonumber\\
\ba{c}
m(x, \lambda)  
 = {  m_{ \kappa_j,0}}(x,\lambda)+  m_{ \kappa_j,r}(x,\lambda),
\ \lambda\in D^\times_{ \kappa_j}, \\ 
 m_{  {\kappa}_1,0}(x, \lambda)= \frac{     \Theta_1(x )}{1-\gamma { \cot^{-1}\frac { \lambda_R-\kappa_1}{|\lambda_I|}}} ,
\\ 
 m_{ {\kappa}_2,0}(x, \lambda)= {\frac{  -\frac{\kappa_1}{\kappa_2}e^{(\kappa_1-\kappa_2)x_1+(\kappa_1^2-\kappa_2^2)x_2-\ln a}\Theta_1(x)}{1-\gamma { \cot^{-1}\frac { \kappa_2-\lambda_R}{|\lambda_I|}}}} , \\ 
|  m_{  {\kappa}_j,0} |_{L^\infty}\le C|v_0|_{L^1\cap L^\infty}, \ \ 
 {  m_{    {\kappa}_j,r}(x, \kappa_j)=0}, \\
 |  m_{ {\kappa}_j,r}  |_{L^\infty\corr{(D_{\kappa_j})}},\   |\frac{\frac\partial{\partial s}  m_{ {\kappa}_j,r} }{1+|x |}|_{L^\infty\corr{(D_{\kappa_j})}}\le C |v_0|_{L^1\cap L^\infty},
\ea\label{E:discon}
\end{gather}with 
\beq\label{E:discrete-coeff}
\begin{array}{c}
\textit{{  $\Theta_{1}(x)=(1+{\mathfrak G}_1\ast v_0 )^{-1}\chi_1\in\RR$}},\\
\textit{{  $\gamma =-\frac 1\pi\iint\xi_1(x )v_0(x )\Theta_1(x)dx\in\RR.$}}
\ea\eeq

\end{theorem}
\begin{proof} $\underline{\emph{Step 1  (Proof of \eqref{E:direct-spectral}-\eqref{E:pole})}}:$ Applying {Proposition} \ref{P:eigen-green} and the assumption $\partial_ x^kv_0 \in  {L^1\cap L^\infty}$, $ 0\le |k|\le 2 $, $|v_0|_{{L^1\cap L^\infty}}\ll 1$, for $\lambda\ne 0$, one can prove  the unique solvability of  the   integral equation 
\beq\label{E:spectral-integral}
\ba{c}
 m(x,\lambda)=\chi(x,\lambda)-  G\ast v_0m (x,\lambda),\\
 m(x,\lambda)\in L^\infty.
\ea
\eeq where the $\ast$ operator is defined by \eqref{E:ast}. 
Besides,  from \eqref{E:sym-0} and Lemma \ref{L:green-heat},  the unique solvability of \eqref{E:direct-spectral} is equivalent to that of \eqref{E:spectral-integral}. 

\corr{  Applying \eqref{E:spectral-integral} and Proposition \ref{P:eigen-green},
\[
\ba{rl}
|(1-E_0)m(x,\lambda)|=|(1+G\ast v_0)^{-1}(1-E_0)\chi|\le C|v_0|_{L^1\cap L^\infty} .
\ea\] So \eqref{E:bdd} is justified. Similarly, for $\lambda\in D_0$, using \eqref{E:sato}, \eqref{E:spectral-integral}, and Proposition \ref{P:eigen-green},
\[
\ba{ l}
m= (1+G\ast v_0)^{-1}(\frac{\chi_0}\lambda+[\chi-\frac{\chi_0}\lambda]).
\ea\]One has
\beq\label{E:g-decomp}
\ba{rl}
m_{res}=& (1+G^o_0\ast v_0)^{-1}\chi_0\\
=&-(1+G^o_0\ast v_0)^{-1}\frac{\kappa_1e^{\theta_1}+\kappa_2ae^{\theta_2}}
{ e^{\theta_1}+ ae^{\theta_2}}\ \in\ L^\infty(D_0),\\
m_{0,r}= &(1+G\ast v_0)^{-1} [\chi-\frac{\chi_0}\lambda]\\
&+ (1+G\ast v_0)^{-1}\frac{G-G^o_0}\lambda\ast v_0(1+G_0^o\ast v_0)^{-1} \chi_0.
\ea
\eeq So \eqref{E:pole} follows.}

$\underline{\emph{Step 2  (Proof of \eqref{E:discon}-\eqref{E:discrete-coeff})}}:$ For $\lambda=\lambda_R+i\lambda_I\in D_{ \kappa_j}^\times$, $j=1,2$, applying \eqref{E:eigen-green}, \eqref{E:spectral-integral}, and defining
\beq\label{E:leading-m-coeff}
\ba{rl}
\wp_j(x,x',\alpha)=&\left\{
{\begin{array}{l}
 1+[{\mathfrak G}_1+\frac 1\pi 
\chi_1(x)\xi_1(x')\cot^{-1}\frac{\lambda_R-\kappa_1}{|\lambda_I|}]\ast  v_0, \\
 1+[{\mathfrak G}_2-\frac 1\pi \chi_2(x)\xi_2(x')\cot^{-1}\frac{\kappa_2-\lambda_R}{|\lambda_I|} ]\ast  v_0,
\end{array}}
\right. 
\\
\Theta_j(x)=&[1+{\mathfrak G}_j(x,x')\ast  v_0(x')]^{-1}\chi_j(x'),\\
\gamma_1 =&-\frac 1\pi\iint\xi_1(x )v_0(x )\Theta_1(x)dx,\\
\gamma_2 =&\frac 1\pi\iint\xi_2(x )v_0(x )\Theta_2(x)dx,
\ea
\eeq
 one has
\[
\ba{c}
m(x,\lambda)=(1+\wp_j^{-1}\omega_j  \ast  v_0)^{-1}\wp_j^{-1}\chi(x,\lambda),
\ea
\] and
\[
\ba{rl}
\blacktriangleright &m_{\kappa_1,0}(x,\lambda)=\wp_1^{-1}\chi_1(x,\lambda)\\
=&(1+[{\mathfrak G}_1 +\frac 1\pi
\chi_1(x) \xi_1 (x')\cot^{-1}\frac{\lambda_R-\kappa_1}{|\lambda_I|}]\ast  v_0)
^{-1}\chi_1\\
=&[1+{\mathfrak G}_1\ast  v_0]^{-1}\chi_1\\
&+  ([1+{\mathfrak G}_1\ast  v_0]^{-1}\frac {-1}\pi
\chi_1(x) \xi_1 (x')\cot^{-1}\frac{\lambda_R-\kappa_1}{|\lambda_I|}\ast  v_0 )[1+{\mathfrak G}_1\ast  v_0]^{-1}\chi_1\\
&+  ([1+{\mathfrak G}_1\ast  v_0]^{-1}\frac {-1}\pi
\chi_1(x) \xi_1 (x')\cot^{-1}\frac{\lambda_R-\kappa_1}{|\lambda_I|}\ast  v_0 )^2[1+{\mathfrak G}_j\ast  v_0]^{-1}\chi_1\\
& +\cdots\\
=&\Theta_1(x)+\gamma_1 \cot^{-1}\frac{\lambda_R-\kappa_1}{|\lambda_I|} \Theta_1 +\left(\gamma_1 \cot^{-1}\frac{\lambda_R-\kappa_1}{|\lambda_I|} \right)^2\Theta_1 +\cdots\\
=& \frac{\Theta_1(x)}{1-\gamma_1  \cot^{-1}\frac{\lambda_R-\kappa_1}{|\lambda_I|} },\\
&\\
\blacktriangleright&m_{\kappa_2,0}(x,\lambda)=\wp_2^{-1}\chi_2(x,\lambda)\\
=&[1+{\mathfrak G}_2\ast  v_0]^{-1}\chi_2\\
&+  ([1+{\mathfrak G}_2\ast  v_0]^{-1}\frac 1\pi 
\chi_2(x) \xi_2 (x')\cot^{-1}\frac{\kappa_2-\lambda_R}{|\lambda_I|}\ast  v_0 )[1+{\mathfrak G}_2\ast  v_0]^{-1}\chi_2
\\
&+  ([1+{\mathfrak G}_2\ast  v_0]^{-1}  \frac 1\pi 
\chi_2(x) \xi_2 (x')\cot^{-1}\frac{\kappa_2-\lambda_R}{|\lambda_I|}\ast  v_0 )^2[1+{\mathfrak G}_2\ast  v_0]^{-1}\chi_2\\
& +\cdots\\
=&\Theta_2+\gamma_2 \cot^{-1}\frac{\kappa_2-\lambda_R}{|\lambda_I|} \Theta_2+\left(\gamma_2 \cot^{-1}\frac{\lambda_R-\kappa_2}{|\lambda_I|} \right)^2\Theta_2 +\cdots\\
=& \frac{\Theta_2(x)}{1-\gamma_2  \cot^{-1}\frac{\kappa_2-\lambda_R}{|\lambda_I|} }.
\ea
\]

To investigate the symmetries between $\Theta_j$ and $\gamma_j$, we combining \eqref{E:green-exp-1-g} with \eqref{E:sym-0},   \eqref{E:eigen-green}, and
 \[
\ba{c}\chi_2(x)=-\frac{\kappa_1}{\kappa_2}\chi_1(x)e^{(\kappa_1-\kappa_2)x _1+(\kappa_1^2-\kappa_2^2)x _2-\ln a},\\
 \ \xi_2(x )=\frac{\kappa_2}{\kappa_1}\xi_1(x )e^{-(\kappa_1-\kappa_2)x _1-(\kappa_1^2-\kappa_2^2)x _2+\ln a},

\ea
\]we obtain
\begin{equation}\label{E:green-exp-2}
\begin{array}{c}
\Theta_2(x)=-\frac{\kappa_1}{\kappa_2}e^{(\kappa_1-\kappa_2)x_1+(\kappa_1^2-\kappa_2^2)x_2-\ln a}\Theta_1(x),\\
\gamma_1=\gamma_2=\gamma 
\end{array}
\end{equation}which,  combining with \eqref{E:arccot}, prove  \eqref{E:discon}.

\end{proof}

\begin{theorem}\label{T:sd-continuous}   
Suppose $\partial_x^kv_0\in  {L^1\cap L^\infty}$, $ | k|\le 2$, $ {|v_0|_{L^1\cap L^\infty}}\ll 1$, and $v_0(x)\in\RR$. Then
\beq\label{E:conti}
\ba{c}
\partial_{\overline\lambda}m(x, \lambda)
=  s_c(\lambda) e^{(\overline\lambda-\lambda)x_1+(\overline\lambda^2-\lambda^2)x_2  }m{(x, \overline\lambda)} , \ \lambda_I\ne 0,
\ea
\eeq with
\beq\label{E:conti-sd}
\ba{rl}
  s_c(\lambda) =&\frac {\mbox{sgn}(\lambda_I)}{2\pi i} \iint e^{-[(\overline\lambda-\lambda)x_1+(\overline\lambda^2-\lambda^2)x_2 ] } \xi(x, \overline\lambda) v_0(x)m(x, \lambda)dx \\
 \equiv&\frac {\mbox{sgn}(\lambda_I)}{2\pi i} \widehat{\xi  v_0 m }(\frac {\overline\lambda-\lambda}{2\pi i}, \frac {\overline\lambda^2-\lambda^2}{2\pi i};\lambda).
 
\end{array}
\eeq  
 Moreover, if $
 \partial_x^k v_0\in {L^1}\cap L^\infty$, $|k|\le 2 $, then
\beq
\ba{c}
| (1-E_{D_{\kappa_1}\cup D_{\kappa_2}}  )  s_c  |_{   L^2(|\lambda_I| d\overline\lambda \wedge d\lambda)\cap L^\infty}  \le  C\sum_{|k|\le 2} |\partial_x^kv_0| _{L^1\cap L^\infty},
\ea\label{E:pm-i-new}  
\eeq
and if $
\corr{(1+|x|)}\partial_x^k v_0\in {L^1}\cap L^\infty$, $|k|\le 2 $, then
\begin{gather}
\ba{c} 
  s_c(\lambda)=
\left\{
{\ba{ll}
 \frac{ \frac {i}{ 2} \mbox{sgn}(\lambda_I)}{\overline\lambda-\kappa_1}\frac {+\gamma}{1-\gamma { \cot^{-1}\frac { \lambda_R-\kappa_1}{|\lambda_I|}}}+\mbox{sgn}(\lambda_I)h_1(\lambda),&\lambda\in D^ \times_{ \kappa_1 },\\
 \frac{ \frac {i}{2}\mbox{sgn}(\lambda_I) }{\overline\lambda-\kappa_2}\frac {-\gamma}{1-\gamma { \cot^{-1}\frac { \kappa_2-\lambda_R}{|\lambda_I|}}}+\mbox{sgn}(\lambda_I) {h_2}(\lambda),&\lambda\in  D^ \times_{ \kappa_2 },
 \\
\mbox{sgn}(\lambda_I) {h_0}(\lambda),&\lambda\in  D^\times _{ 0},
\ea}
\right.
\ea\label{E:cd-decomposition} 
\end{gather}
where 
   $E_{z,a}$, $D_z^\times$, $\cot^{-1}\frac{ {\kappa_2}- \lambda_R}{|\lambda_I|}$, $\cot^{-1}\frac{  \lambda_R-\kappa_1}{|\lambda_I|}$, $\gamma $ are defined by  Definition \ref{D:terminology}, \eqref{E:arccot}, and \eqref{E:discrete-coeff}. Moreover,
\beq \label{E:cd-decomposition-new}
\ba{c}
 |\gamma|_{L^\infty}\le |v_0|_{L^1}, \ 
  \corr{  { |\sum_{0\le |l|\le 1} \partial_{\lambda_R}^{l_1}\partial_{\lambda_I}^{l_2}h_k|_{L^\infty}\le C |(1+|x|)     v_0|_{L^1\cap L^\infty}, }  } 
 \\
 h_k(\lambda)=\overline{h_k( \overline\lambda)} .\ea 
\eeq  
  
\end{theorem}

\begin{proof}   
$\underline{\emph{Step 1 (Proof of \eqref{E:conti})}}:$   Denote   $\rho(x,  \lambda )= e^{(\overline\lambda-\lambda)x_1+(\overline\lambda^2-\lambda^2)x_2  }$. Note $\rho(x, \lambda )$ is annihilated by the heat operator $p_\lambda(D)\equiv -\partial_{x_2}+\partial_{x_1}^2+2 \lambda\partial_{x_1}$. So $
  p_\lambda(D)f=e^{(\overline\lambda-\lambda)x_1+(\overline\lambda^2-\lambda^2)x_2  }p_{ \overline \lambda}(D)e^{-[(\overline\lambda-\lambda)x_1+(\overline\lambda^2-\lambda^2)x_2 ] }f
$ which yields 
\begin{equation}\label{E:green-exp}
\ba{c}
  G_\lambda\,  \rho(x,  \lambda )=\rho(x,  \lambda) \,   G_{ \overline\lambda}.
\ea
\end{equation}
Therefore, for  $\lambda_I\ne 0$,   by Lemma \ref{L:cont-debar},  \eqref{E:spectral-integral}, \eqref{E:conti-sd}, and \eqref{E:green-exp},
\[
\ba{rl}
&\partial_{\overline\lambda}m {(x,\lambda)}\\
=& \partial_{\overline\lambda}\left[(1+G_\lambda\ast v_0)^{-1}\chi\right] \\
=&- (1+ G_\lambda\ast v_0)^{-1}\left(\partial_{\bar\lambda}  G_\lambda\ast v_0\right)m {(x,\lambda)} \\
=& 
-(1+ G_\lambda\ast v_0)^{-1}\frac {\texttt{sgn}(\lambda_I)\rho(x-x',\lambda )\chi(x, \overline\lambda)\xi(x', \overline\lambda)}{-2\pi i} \ast v_0m \\
=&  s_c(\lambda)(1+ G_\lambda\ast v_0)^{-1}\rho(x,\lambda)\chi(x, \overline\lambda) \\
=&  s_c(\lambda) \rho(x,\lambda)(1+G_{ \overline\lambda}\ast v_0)^{-1}\chi(x, \overline\lambda) \\
=&  s_c(\lambda) e^{(\overline\lambda-\lambda)x_1+(\overline\lambda^2-\lambda^2)x_2  }m{(x, \overline\lambda)}.
\ea
\]

$\underline{\emph{Step 2  (Proof of \eqref{E:cd-decomposition}, \eqref{E:cd-decomposition-new})}}:$ From    \eqref{E:sato-adjoint},  \eqref{E:discon}-\eqref{E:discrete-coeff}, and \eqref{E:conti-sd},
\[
\ba{rl}
  &s_c(\lambda) \\
=&\frac {\mbox{sgn}(\lambda_I)}{2\pi i} \iint e^{-[(\overline\lambda-\lambda)x_1+(\overline\lambda^2-\lambda^2)x_2 ] }\xi(x, \overline\lambda) v_0(x)m(x, \lambda)dx 
\\
=  &\left\{
{\ba{l }
\frac{\frac {\mbox{sgn}(\lambda_I)}{2\pi i}}{\overline\lambda-\kappa_1} \iint   \xi_1 (x)v_0(x)\frac{     \Theta_1(x )}{1-\gamma { \cot^{-1}\frac { \lambda_R-\kappa_1}{|\lambda_I|}}}dx+\mbox{sgn}(\lambda_I)h_1(\lambda),    \lambda\in D_{\kappa_1}^\times , \\
\frac{\frac {\mbox{sgn}(\lambda_I)}{2\pi i}}{\overline\lambda-\kappa_2} \iint   \xi_2 (x)v_0(x)\frac{     \Theta_2(x )}{1-\gamma { \cot^{-1}\frac {\kappa_2- \lambda_R}{|\lambda_I|}}}dx+\mbox{sgn}(\lambda_I)h_2(\lambda), \lambda\in D_{\kappa_2}^\times ,\\
\frac{\frac {\mbox{sgn}(\lambda_I)}{2\pi i}}{\lambda} \iint   \xi  (x,0)v_0(x)m_{res}(x)dx+\mbox{sgn}(\lambda_I)h_0(\lambda), \quad\quad\quad \lambda\in D_{0}^\times
\ea}
\right.\\
=   &\left\{
{\ba{ll}
-\frac{\frac {\mbox{sgn}(\lambda_I)}{2  i} \frac \gamma{1-\gamma { \cot^{-1}\frac { \lambda_R-\kappa_1}{|\lambda_I|}}}}{\overline\lambda-\kappa_1}+\mbox{sgn}(\lambda_I)h_1(\lambda), & \lambda\in D_{\kappa_1}^\times \\
+\frac{\frac {\mbox{sgn}(\lambda_I)}{2 i} \frac \gamma{1-\gamma { \cot^{-1}\frac { \kappa_2-\lambda_R}{|\lambda_I|}}}}{\overline\lambda-\kappa_2}+\mbox{sgn}(\lambda_I)h_2(\lambda), & \lambda\in D_{\kappa_2}^\times , \\
\mbox{sgn}(\lambda_I)h_0(\lambda),& \lambda\in D_{0}^\times,

\ea}
\right.
\ea
\]\corr{which is \eqref{E:cd-decomposition} and estimates for \eqref{E:cd-decomposition-new} can be derived directly.}
 
$\underline{\emph{Step 3  (Proof of \eqref{E:pm-i-new})}}:$ Using a similar argument as in $\underline{\emph{Step 2}}$, one can prove $
|E_0s_c|_{L^\infty}\le C$ as well. Moreover, from  \eqref{E:conti-sd}, the Fourier theory, and Theorem \ref{T:KP-eigen-existence},   
 \[
 \ba{rl}
 &\iint{|(1-E_{D_{\kappa_1}\cup D_{\kappa_2}} (\lambda) )  s_c(\lambda)|^2}|\lambda_I|d\overline\lambda\wedge d\lambda\\
 = &\iint{|(1-E_{D_{\kappa_1}\cup D_{\kappa_2}} (\lambda) )   \frac {\mbox{sgn}(\lambda_I)}{2\pi i} \widehat{\xi  v_0 m }(\frac {\overline\lambda-\lambda}{2\pi}, \frac {\overline\lambda^2-\lambda^2}{2\pi};\lambda)|^2}|\lambda_I|d\overline\lambda\wedge d\lambda\\
 \le\ &C\iint \frac{\sum_{k ,j\le 2} |\partial_{x_1}^{k}\partial_{x_2}^{j}v_0|^2_{L^1\cap L^\infty}}{(1+|\lambda_I|^2+|\lambda_R\lambda_I|^2)^2}|\lambda_I|d\overline\lambda\wedge d\lambda\\
 \le & C  \sum_{k ,j\le 2} |\partial_{x_1}^{k}\partial_{x_2}^{j}v_0|^2_{L^1\cap L^\infty}.
  \ea
\]
 
\end{proof}

\corr{Based on the characterization of the eigenfunction $m$, we define the eigenfunction space $W$ and the spectral transformation $T$ in Definition \ref{D:quadrature-hat} and \ref{D:spectral}.}
{\begin{definition}\label{D:quadrature-hat}
  The  eigenfunction space ${ W }\equiv { W}_{x}$ is the set of functions 
 \[
\begin{array}{rl}
i. & \phi (x, \lambda)=\overline{ \phi (x, \overline\lambda)};\\
ii. & (1-  E_0 )\phi(x, \lambda)\in L^\infty;\\
iii. & \phi(x, \lambda)=\frac{\phi_{res}(x )}{\lambda }  +\phi_{0,r}(x, \lambda),\ \lambda \in D_0^\times,
\\
&\phi_{res}(x ), \ \corr{\lambda}\phi_{0,r}(x ,\lambda),\ 
   \frac{\phi_{0,r}(x, \lambda)}{1+|x| } \in L^\infty(D_0) ;
\\
iv. & \corr{\phi (x, \lambda_2)=s_de^{(\kappa_1-\kappa_2)x_1+(\kappa_1^2-\kappa_2^2)x_2}\phi (x, \lambda_1) },\\
  &  
\corr{\lambda_2= \kappa_2+0^+e^{i\alpha},\  
 \lambda_1= \kappa_1+0^+e^{i(\pi+\alpha)},} \ s_d=-\frac{\kappa_1}{\kappa_2}e^{-\ln a},\\
& \phi(x,\lambda) , \  \frac{\frac\partial{\partial s} \phi_{  \kappa_j,r}(x,\lambda) }{1+|x| } \in L^\infty(D_{\kappa_j}).
\ea
\]
\end{definition} 

\begin{definition}\label{D:spectral} 
Define $\{0;\kappa_1,\kappa_2, s_d,    s _c(\lambda)\}$ as the set of scattering data, 
where $0$, location of the simple pole, $\kappa_j$, location of discontinuities, and $s_d\equiv-\frac {\kappa_1}{\kappa_2}e^{-\ln a}$, the norming constant,   are the {\bf \sl discrete scattering data};  and  $ {  s}_c(\lambda)$, the  {\bf \sl continuous scattering data}, is defined by \eqref{E:conti-sd}. Denote   $T$ as   the {\bf\sl forward scattering transform} by
\beq\label{E:cauchy-operator}
\ba{c}
T(\phi)(x, \lambda)  = {  s}_c(\lambda) e^{(\overline\lambda-\lambda)x_1+(\overline\lambda^2-\lambda^2)x_2  }\phi(x, \overline\lambda).
\ea
\eeq 
\end{definition}

\begin{definition}\label{D:cauchy}
Let   $\mathcal C$ be the Cauchy integral operator  defined by 
\beq\label{E:cauchy-new}
\ba{c}
\mathcal C(\phi)(x,  \lambda)=\mathcal C_\lambda(\phi) =-\frac 1{2\pi i}\iint\frac {\phi(x, \zeta)}{\zeta-\lambda}d\overline\zeta\wedge d\zeta. 
\ea
\eeq 
 
\end{definition}

\corr{We now provide estimates on the spectral transform of $m$ in order to formulate a Cauchy integral equation in Section \ref{S:cauchy}.}
\begin{theorem}\label{T:bdd} Suppose $
  \partial_x^k v_0\in {L^1}\cap L^\infty$, $|k|\le 2 $, $
|v_0|_{L^1\cap L^\infty}\ll 1$.   Then   
\[
\ba{rl}
&|\mathcal C   T m| _{L^\infty} \le C   (1+|x|)  \sum_{|k|\le 2}| \partial_x^k v_0|_{L^1\cap L^\infty}
 ,\\
 & \mathcal C T m (x, \lambda)\to 0 ,\quad\textit{ as $|\lambda|\to\infty$, $\lambda_I\ne 0$ .} 
\ea
\] 
\end{theorem}

\begin{proof}   $\underline{\emph{Step 1  (Near $z\in \{\kappa_1,\kappa_2\}$)}}:$ From \eqref{E:conti}, applying Stokes' theorem,
\beq\label{E:stokes-1}
\ba{rl}
&-\frac 1{2\pi i}\iint_{D_{z}/\RR\cup(D_{z,\epsilon }\cup D_{\lambda,\epsilon})}\frac{s_{c }(\zeta)e^{(\overline\zeta-\zeta)x_1+(\overline\zeta^2-\zeta^2)x_2}m(x, \overline\zeta)} {\zeta-\lambda}d\overline  \zeta\wedge d\zeta\\
=&-\frac 1{2\pi i}\iint_{D_{z}/\RR\cup(D_{z,\epsilon }\cup D_{\lambda,\epsilon})}\frac{\partial_{\overline\zeta}m(x,  \zeta)} {\zeta-\lambda}d\overline  \zeta\wedge d\zeta\\
=&-\frac 1{2\pi i}\int_{\partial\left[D_{ z}/(\RR\cup D_{z,\epsilon }\cup D_{\lambda,\epsilon})\right]} \frac{m(x,\zeta)}{\zeta-\lambda}d\zeta\\
=&-\frac {1} {2\pi i}\oint_{|\zeta-z|=\corr{\kappa}}\frac{  m(x,  \zeta)}{\zeta- \lambda}d\zeta+\frac 1{2\pi i}\int_{\partial D_{z,\epsilon } } \frac{m(x, \zeta)}{\zeta-\lambda}d\zeta +\frac 1{2\pi i}\int_{\partial D_{\lambda,\epsilon} } \frac{m(x,\zeta)}{\zeta-\lambda}d\zeta 
\ea
\eeq
  where $\corr{\kappa=\frac 12\min\{|\kappa_1|,\,|\kappa_2|,\,\kappa_2-\kappa_1\}}$ is defined by Definition \ref{D:terminology}. Note,  by $\lambda\ne \kappa_j$, \eqref{E:discon}, and \eqref{E:cd-decomposition},
\beq\label{E:stokes-2}
\ba{rl}
&-\frac 1{2\pi i}\iint_{D_{z,\epsilon }}\frac{s_{c}(\zeta)e^{(\overline\zeta-\zeta)x_1+(\overline\zeta^2-\zeta^2)x_2}m (x, \overline\zeta)} {\zeta-\lambda}d\overline  \zeta\wedge d\zeta\ \to 0,\\
&-\frac 1{2\pi i}\iint_{D_{\lambda,\epsilon}}\frac{s_{c}(\zeta)e^{(\overline\zeta-\zeta)x_1+(\overline\zeta^2-\zeta^2)x_2}m (x, \overline\zeta)} {\zeta-\lambda}d\overline  \zeta\wedge d\zeta\ \to 0,\\
&+\frac 1{2\pi i}\int_{\partial D_{z,\epsilon } } \frac{m(x,  \zeta)}{\zeta-\lambda}d\zeta\ \to 0,\\
&+\frac 1{2\pi i}\int_{\partial D_{\lambda,\epsilon} } \frac{m(x,  \zeta)}{\zeta-\lambda}d\zeta\ \to m(x,  \lambda),\ \ \textit{as $\epsilon\to 0$}.
\ea
\eeq Therefore 
\beq\label{spectral-kj}
\ba{rl}
&|\mathcal CTE_{z}m|_{L^\infty}\\
=&|-\frac 1{2\pi i}\iint \frac{ E_{ \kappa_j}s_{c}(\zeta)e^{(\overline\zeta-\zeta)x_1+(\overline\zeta^2-\zeta^2)x_2}m (x, \overline\zeta)} {\zeta-\lambda}d\overline  \zeta\wedge d\zeta|_{L^\infty}\\
=&|m(x,\lambda)-\frac {1} {2\pi i}\oint_{|\zeta-z|=\corr{\kappa}}\frac{  m(x,  \zeta)}{\zeta- \lambda}d\zeta|_{L^\infty}\\
\le &C|v_0|_{L^1\cap L^\infty}.

\ea
\eeq 

 $\underline{\emph{Step 2  (Near $0$)}}:$ From \eqref{E:conti},  
 \[
 \ba{l}
 \partial_{\overline\lambda}\left(m(x, \lambda)-\frac {m_{res}(x)}\lambda\right)
=  s_c(\lambda) e^{(\overline\lambda-\lambda)x_1+(\overline\lambda^2-\lambda^2)x_2  }m{(x, \overline\lambda)} , \ \lambda\in D_0/\RR.
 \ea\] Applying Stokes' theorem,
\[
\ba{rl}
&-\frac 1{2\pi i}\iint_{D_{0}/\RR\cup(D_{0,\epsilon }\cup D_{\lambda,\epsilon})}\frac{s_{c }(\zeta)e^{(\overline\zeta-\zeta)x_1+(\overline\zeta^2-\zeta^2)x_2}m(x, \overline\zeta)} {\zeta-\lambda}d\overline  \zeta\wedge d\zeta\\
=&-\frac 1{2\pi i}\iint_{D_{0}/\RR\cup(D_{0,\epsilon }\cup D_{\lambda,\epsilon})}\frac{\partial_{\overline\zeta}\left(m(x,  \zeta)-\frac{m_{res}(x)}\zeta\right)} {\zeta-\lambda}d\overline  \zeta\wedge d\zeta\\
=&-\frac 1{2\pi i}\int_{\partial\left[D_{0}/(\RR\cup D_{0,\epsilon }\cup D_{\lambda,\epsilon})\right]} \frac{m_{0,r}(x,\zeta)}{\zeta-\lambda}d\zeta\\
=&-\frac {1} {2\pi i}\oint_{|\zeta|=\corr{\kappa}}\frac{   m_{0,r}(x,  \zeta)}{\zeta- \lambda}d\zeta+\frac 1{2\pi i}\int_{\partial D_{0,\epsilon } } \frac{   m_{0,r}(x, \zeta)}{\zeta- \lambda}d\zeta\\
&+\frac 1{2\pi i}\int_{\partial D_{\lambda,\epsilon} } \frac{   m_{0,r}(x,\zeta)}{\zeta- \lambda}d\zeta 
\ea
\]Note,  for $x\in\RR^2$, $\lambda\ne 0$ fixed, \eqref{E:pole}, and \eqref{E:cd-decomposition}, 
\beq\label{E:stokes-2-0}
\ba{rl}
&-\frac 1{2\pi i}\iint_{D_{0,\epsilon }}\frac{s_{c}(\zeta)e^{(\overline\zeta-\zeta)x_1+(\overline\zeta^2-\zeta^2)x_2}m (x, \overline\zeta)} {\zeta-\lambda}d\overline  \zeta\wedge d\zeta\ \to 0,\\
&-\frac 1{2\pi i}\iint_{D_{\lambda,\epsilon}}\frac{s_{c}(\zeta)e^{(\overline\zeta-\zeta)x_1+(\overline\zeta^2-\zeta^2)x_2}m (x, \overline\zeta)} {\zeta-\lambda}d\overline  \zeta\wedge d\zeta\ \to 0,\\
&+\frac 1{2\pi i}\int_{\partial D_{0,\epsilon } } \frac{  m_{0,r}(x,  \zeta)}{\zeta- \lambda}d\zeta\ \to 0,\\
&+\frac 1{2\pi i}\int_{\partial D_{\lambda,\epsilon} } \frac{  m_{0,r}(x,  \zeta)}{\zeta- \lambda}d\zeta\ \to m_{0,r}(x,  \lambda),\ \ \textit{as $\epsilon\to 0$}.
\ea
\eeq Therefore 
\beq\label{spectral-kj-0}
\ba{rl}
&|\mathcal CTE_{0}m|_{L^\infty}\\
=&|-\frac 1{2\pi i}\iint \frac{ E_{ 0}s_{c}(\zeta)e^{(\overline\zeta-\zeta)x_1+(\overline\zeta^2-\zeta^2)x_2}m (x, \overline\zeta)} {\zeta-\lambda}d\overline  \zeta\wedge d\zeta|_{L^\infty}\\
=&|m_{0,r}(x,\lambda)-\frac {1} {2\pi i}\oint_{|\zeta|=\corr{\kappa}}\frac{  m_{0,r}(x,  \zeta)}{\zeta- \lambda}d\zeta|_{L^\infty}\\
\le &C(1+|x|)|v_0|_{L^1\cap L^\infty}.

\ea
\eeq 

 $\underline{\emph{Step 3   (Near $\infty$)}}:$ The proof can be applied to $\forall \phi\in W$. Via a change of variables 
 \beq\label{E:wickhauser-spec}
 \ba{c}
2\pi i\xi= \overline\zeta-\zeta,\quad 2\pi i\eta=\overline\zeta^2-\zeta^2,\\
\zeta=-i\pi\xi+ \frac \eta{2\xi}, \, d\overline\zeta\wedge d\zeta=\frac{i\pi}{|\xi|}d\xi d\eta,
\ea
\eeq 
and from \eqref{E:pm-i-new},    \cite[Lemma 2.II]{W87}
\begin{equation}\label{E:wick-infty-spec}
\ba{c}
p_\lambda(\xi,\eta)=(2\pi\xi)^ 2-4\pi i\xi\lambda+2\pi i \eta ,\ \Omega_\lambda=\{(\xi,\eta)\in\RR^2\ :\ |p_\lambda(\xi,\eta)|<1\},\\
\left|\frac 1{p_\lambda}\right|_{L^1(\Omega_\lambda, d\xi d\eta)}\le  \frac C{(1+|\lambda_I|^2)^{1/2}},\ 
\left|\frac 1{p_\lambda}\right|_{L^2(\Omega_\lambda^c,d\xi d\eta)}\le \frac C{(1+|\lambda_I|^2)^{1/4}},
\ea
\end{equation} along with \eqref{E:bdd} and \eqref{E:pm-i-new},  we obtain
\beq\label{E:infty-spec}
\begin{array}{rl}
&|\mathcal C[1-  E_{D_0\cup D_{\kappa_1}\cup D_{\kappa_2}} ] T\phi|\\
\le & C|   \iint \frac {{ \left[1-E_{D_{0}\cup D_{\kappa_1}\cup D_{\kappa_2}}(\zeta) \right]  s_c(\zeta)e^{(\overline\zeta-\zeta)x_1+(\overline\zeta^2-\zeta^2)x_2}\phi}}{ \zeta-\lambda }d\overline\zeta\wedge d\zeta| \\
\le & C| [(1-  E_0 )  ] \phi|_{L^\infty} \iint \frac {|{ \left[1-E_{D_{\kappa_1}\cup D_{\kappa_2}}(\zeta)\right]\mathfrak s_c(\zeta)|}}{|(2\pi\xi)^ 2-4\pi i\xi\lambda+2\pi i \eta|}d\xi d\eta
\\
\le &C| [(1-  E_0 )  ]\phi|_{L^\infty}  \{|\left[1-E_{D_{\kappa_1}\cup D_{\kappa_2}}(\zeta)\right]  s_c(\zeta)| _{L^2(d\xi d\eta)}\left|\frac 1{p_\lambda}\right|_{L^2(\Omega^c_\lambda,d\xi d\eta)} \\
& +{|\left[1-E_{D_{\kappa_1}\cup D_{\kappa_2}}(\zeta)\right]  s_c(\zeta) |}_{L^\infty(d\xi d\eta)}\left|\frac 1{p_\lambda}\right|_{L^1(\Omega_\lambda, d\xi d\eta)}\}\\
\le &C 

\end{array}
\eeq 
\corr{and tends to $0$ as $|\lambda|\to\infty$, $\lambda_I\ne 0$.}
 
\end{proof}

We make several remarks about Theorem \ref{T:bdd} since it is necessary for the justification of a Cauchy integral equation for $m$. 
\begin{itemize}
\item \corr{Due to \eqref{E:pm-i-new},  there is a {missing direction in the $\lambda$-plane} (the real axis)   for  $  s_c(\lambda)$ to decay no matter how smooth the initial data $v_0(x)$ is. Therefore, boundedness of   {$m(x,\lambda)  $} on $\lambda\in D_0^c\cap D_{\kappa_1}^c\cap D_{\kappa_2}^c$    is vital in deriving uniform estimates there.  
\item 
  Boiti-Pempinelli-Pogrebkov's eigenfunction \cite{BP214}, defined by  
\[
\ba{c}
{   \lambda m(x, \lambda)  } ,
\ea
\]which is not bounded near $\infty$, cannot be    admissible. }
 \item  \corr{In \cite{Wu18}, via a KdV approach, 
the boundary condition, i.e., the Sato eigenfunction $\varphi(x,\lambda)$, is replaced by the KdV eigenfunction $\corr{\psi_-(x, \lambda)}$
\[
\ba{c}
\psi_-(x, \lambda)= (1+ \frac {  2i  {\kappa}} {\lambda- i {\kappa}} \frac {1} {1 +e^{2 {\kappa}x_1}} ) e^{(-i\lambda) x_1+(-i\lambda)^2x_2}.
\ea
\]
\begin{itemize}
\item So   the eigenfunction for the boundary value problem of the Lax equation is 
\[\ba{c}{  \frac{\lambda}{\lambda- {i\kappa }}\tilde m(x, \lambda),\quad\tilde m\in L^\infty  },\ea\]whileas $i\kappa$ is also a singularity of the corresponding $\tilde s_c$. Therefore, a standard Cauchy Theorem (shown in the proof of \eqref{E:stokes-2}) will collapse at $D_{{i\kappa },\epsilon}$, $\partial D_{{i\kappa },\epsilon}$.
\item We then introduced a {regularization at ${i\kappa }$}
\beq\label{E:kdv-reg}
\ba{c}{{(\lambda-{i\kappa })} \frac{\lambda}{\lambda- {i\kappa }}\tilde m(x, \lambda)  } \ea\eeq to remedy the problem near ${i\kappa }$. But \eqref{E:kdv-reg} becomes unbounded at $\infty$.

\item Finally, we introduce an {extra pole $2{i\kappa }$} to tame the singularities at   $\infty$  and obtain
\beq\label{E:kdv-eigen}\ba{c}{  \frac{{\lambda- {i\kappa }}}{{\lambda-2{i\kappa }}}\frac{\lambda}{\lambda- {i\kappa }}\tilde m(x, \lambda)  }.\ea\eeq 
\end{itemize}   } 
    
\corr{The eigenfunction \eqref{E:kdv-eigen} is    admissible.}

\end{itemize}


\section{The Cauchy integral equation}\label{S:cauchy}

\begin{theorem}\label{T:cauchy-integral-eq}
If
\[
\ba{c}
u_0(x )=\frac{(\kappa_1-\kappa_2)^2}2\textrm{sech}^2\frac{\theta_1-\theta_2-\ln a}2,\\
\partial_x^k v_0\in {L^1}\cap L^\infty,\ |k|\le 4 ,\\
|v_0 |_{{L^1\cap L^\infty}}\ll 1, \ v_0(x)\in\RR,
\ea
\] then the eigenfunction $  m $  derived from Theorem \ref{T:KP-eigen-existence}  satisfies 
\beq\label{E:eigen-charac}
\ba{c}
 m (x, \lambda)\in W
 \ea
 \eeq and the Cauchy integral equation
\beq\label{E:cauchy-integral-equation-sf}
\begin{array}{cl}
  {  m}(x, \lambda) =1+\frac{  m_{ res }(x  )}{\lambda  }  +\mathcal CT
   m , & \forall\lambda\ne  0,
\end{array}
\eeq  where $W$ is defined by   Definition \ref{D:quadrature-hat}.

\end{theorem}

\begin{proof}   
Theorem \ref{T:KP-eigen-existence} implies, for $x$ fixed, 
\begin{gather}
\ba{c}
  m(x, \lambda)-\frac{ m_{  res }(x   )}{\lambda } \quad\in L^\infty,\ea \label{E:regularize-m-ifty}\\
\ba{c}
  E_{0,n} T  m(x, \lambda)\in  L^1(d\overline\lambda\wedge d\lambda), 
\ea\label{E:t-regularize-out-sing}
\end{gather} for $\forall   n>0$. Here $E_{z,a}$ is defined by Definition \ref{D:terminology}. 
Exploiting \eqref{E:t-regularize-out-sing} and applying \cite[\S I, Theorem 1.13, Theorem 1.14]{V62}, one 
derives
\beq\label{E:l-1-dbar}
\ba{c}
\partial_{\overline\lambda}\mathcal CE_{0,n} T  m(x, \lambda) 
 =E_{0,n} T  m(x, \lambda) \in  L^1(d\overline\lambda\wedge d\lambda).
 \ea
\eeq  Therefore, together with Theorem \ref{T:sd-continuous}, Theorem \ref{T:bdd},
\beq\label{E:debar-initial}
\ba{c}
\partial_{\overline\lambda}\left[ m(x, \lambda)-\frac{  m_{ res }(x  )}{\lambda }-\mathcal CT m(x, \lambda)\right]=0.
\ea
\eeq
For $\forall  x$ fixed, applying Theorem \ref{T:bdd}, \eqref{E:regularize-m-ifty}, \eqref{E:debar-initial},    and  Liouville's theorem,  one concludes  
\beq\label{E:unique-direct-problem-q}
\ba{c}
  m(x, \lambda)=g(x )+\frac{  m_{ res }(x )}{\lambda }  +\mathcal CT  m(x ,\lambda).
\ea
\eeq  

Equation \eqref{E:direct-spectral} and a direct computation  yield:
\beq\label{E:unique-computation}
\ba{rl}
&u(x )  m(x,  \lambda)\\
=&\left(\partial_{x_2}-\partial_{x_1}^2-2 \lambda\partial_{x_1}\right)  m(x,\lambda)\\
=&\left(\partial_{x_2}-\partial_{x_1}^2-2 \lambda\partial_{x_1}\right)
[g(x )+\frac    {\mathfrak m_{ res }(x )}{\lambda }    ]\\
&+\left(\partial_{x_2}-\partial_{x_1}^2-2 \lambda\partial_{x_1}\right)\mathcal CT  m.
\ea
\eeq
Note that
\begin{eqnarray*}
&&\partial_{x_1}\mathcal CT  m=\mathcal C[ (\overline\lambda-\lambda)T  m+T(\partial_{x_1}   m)],\\
&&\partial_{x_1}^2\mathcal CT  m=\mathcal C[(\overline\lambda -\lambda  )^2T  m+2 (\overline\lambda-\lambda)T(\partial_{x_1}   m)+T(\partial_{x_1} ^2  m)],\\
&&\partial_{x_2}\mathcal CT  m=\mathcal C[ (\overline\lambda^2- \lambda^2) T  m +T(\partial_{x_2}  m)].
\end{eqnarray*} 
Applying the Fourier transform theory, if $v_0(x)$ has $4$ derivatives in  $L^1\cap L^\infty$, then  
 \[
 \ba{l}
 (1-E_{D_{\kappa_1}\cup D_{\kappa_2}}(\lambda))(\overline\lambda-\lambda)  s_c(\lambda),\\
 (1-E_{D_{\kappa_1}\cup D_{\kappa_2}}(\lambda))(\overline\lambda-\lambda)^2   s_c(\lambda),\\
 (1-E_{D_{\kappa_1}\cup D_{\kappa_2}}(\lambda))(\overline\lambda^2- \lambda^2)   s_c(\lambda),
 \ea
 \] 
are all bounded in $L^\infty\cap L^2(|\lambda_I|d\overline\lambda\wedge d\lambda)$. Therefore, if  $\partial_{x }^kv_0 \in  {L^1\cap L^\infty}$, $ 0\le |k|\le 4 $,  one can adapt the proof of \corr{$\underline{\emph{Step 1}}$} - $\underline{\emph{Step 3}}$ in Theorem \ref{T:bdd} and  derive, as $|\lambda|\to\infty$, $\lambda_I\ne 0$, 
\[ 
\ba{c}
\left(\partial_{x_2}-\partial_{x_1}^2-2 \lambda\partial_{x_1}\right)\mathcal CT  m\to o(|\lambda|).
\ea
\] 
So 
comparing growth in \eqref{E:unique-computation}, we conclude
  \eqref{E:unique-direct-problem-q} turns into
\beq\label{E:lambda-infty}
\ba{c}
  m(x, \lambda)-1=g(x_2)-1+\frac    { m_{ res }(x )}{\lambda }+ \mathcal C T   m(x, \lambda).
\ea
\eeq
Fix $x_2$, and let $\epsilon$ be given. \corr{Let  $x_1\gg 1$, $|\lambda|\gg 1$, $\lambda_I\ne 0$, one has
\[
|\frac    {   m_{ res }(x )}{\lambda} +\mathcal C T    m(x, \lambda)|\le \frac\epsilon 2
\]by Theorem \ref{T:bdd} and 
\[
\ba{c}
|m(x,\lambda)-1|\le\frac\epsilon 2
\ea\] by the boundary property   \eqref{E:direct-spectral}.} So we justify  $g\equiv 1$ and establish \eqref{E:cauchy-integral-equation-sf}.


\end{proof}

Theorem \ref{T:cauchy-integral-eq} implies that  
   the residue $  m_{ res }(x  )$ at the simple pole $\lambda=0$ and   $  m({\color{red} x},\kappa_j+0^+e^{i\alpha})$ at $\kappa_j$ satisfy the linear system
\beq\label{E:+i-pogrebkov-x-function-new}
\ba{c}
 \frac{  m_{ res }(x  )}{ \kappa_1}  =-1+  m (x, \kappa_1+0^+e^{i(\pi+\alpha)})  - \mathcal C_{ \kappa_1+0^+e^{i(\pi+\alpha)}} T  m  , \\
 \frac{  m_{ res }(x  )}{ \kappa_2}  =-1+  m (x, \kappa_2+0^+e^{i\alpha}) - \mathcal C_{ \kappa_2+0^+e^{i\alpha}}T  m ,\\
   m (x,  \kappa_2+0^+e^{i\alpha})=  s_de^{(\kappa_1-\kappa_2)x_1+(\kappa_1^2-\kappa_2^2)x_2} m (x, \kappa_1+0^+e^{i(\pi+\alpha)})
\ea
\eeq for $\forall 0<\alpha <2\pi $, with  
   $T$, $  s_d$, and $\mathcal C$ defined by Definition \ref{D:spectral}.

\begin{example}\label{Ex:cauchy-eq-v-0}
If   $  s_c\equiv 0$, $s_d=-\frac {\kappa_1}{\kappa_2}e^{-\ln a}$. So \eqref{E:cauchy-integral-equation-sf} and \eqref{E:+i-pogrebkov-x-function-new} reduce to
\[\ba{c}
 {  m}(x, \lambda) =1+\frac{  m_{ res }(x   )}{\lambda } ,\\
    m (x, \kappa_2 )=  {  s}_de^{(\kappa_1-\kappa_2)x_1+(\kappa_1^2-\kappa_2^2)x_2} m (x,\kappa_1)   
\ea\] which yield
\begin{gather}
\ba {c}\frac{  m_{ res }(x   )}{ {\kappa_1}}  =-1+  m (x,\kappa_1)      ,\ea\label{E:2} \\
 \ba{c}\frac{  m_{ res }(x  )}{ \kappa_2} =-1+ m (x,\kappa_2)  ,\ea\label{E:3}\\
 \ba{c}  m (x,\kappa_2)  =  {  s}_de^{(\kappa_1-\kappa_2)x_1+(\kappa_1^2-\kappa_2^2)x_2} m (x,\kappa_1)  
\ea\label{E:4}
\end{gather} Namely,
\[
\ba{c}
   m (x, \kappa_1  )=+\frac{\kappa_1-\kappa_2}{\kappa_1}\frac 1{1+e^{ (\kappa_1-\kappa_2)x_1+(\kappa_1^2-\kappa_2^2)x_2-\ln a }},\\ 
   m (x, \kappa_2 )=-\frac{\kappa_1-\kappa_2}{\kappa_2}\frac 1{1+e^{-[ (\kappa_1-\kappa_2)x_1+(\kappa_1^2-\kappa_2^2)x_2-\ln a] }}, 
   \ea\] and
\[\ba{rl}   
 &m(x,\lambda)\\
 =&
1-\frac 1{\lambda }( \frac{\kappa_1}{1+e^{-[ (\kappa_1-\kappa_2)x_1+(\kappa_1^2-\kappa_2^2)x_2-\ln a] }}+\frac{\kappa_2}{1+e^{ (\kappa_1-\kappa_2)x_1+(\kappa_1^2-\kappa_2^2)x_2-\ln a  }})\\
=&\chi(x,\lambda). 
\ea
\]
\end{example}


\begin{thebibliography}{9}

\bibitem{AC91}
M.J. Ablowitz, P.A. Clarkson: 
Solitons, nonlinear evolution equations and inverse scattering. \textit{London Mathematical Society lecture note series, }149   (1991), Cambridge University Press.

\bibitem{BPPP01-p} 
M. Boiti, F. Pempinelli, A. Pogrebkov, B. Prinari: Inverse scattering transform for the perturbed 1-soliton potential of the heat equation. \textit{Phys. Lett. A}  \textbf{285}  (2001),  no. 5-6, 307-311.

\bibitem{BP302}
M. Boiti, F. Pempinelli, A. K. Pogrebkov, B. Prinari: Inverse scattering theory of the heat equation for a perturbed one-soliton potential. \textit{J. Math. Phys.}{\bf 43} (2002), no. 2, 1044-1062.

\bibitem{BP211}
M. Boiti, F. Pempinelli, A. K. Pogrebkov:
Green's function of heat operator with pure soliton potential. \textit{ArXiv:1201.0152v1}, 1-10,[nlin.SI] 30 Dec 2011.
 
 
 

\bibitem{BP214}
M. Boiti, F. Pempinelli, A. K. Pogrebkov:
 IST of KPII equation for perturbed multisoliton solutions. \textit{Topology, geometry, integrable systems, and mathematical physics}, 49-73, Amer. Math. Soc. Transl. Ser. 2, 234, Amer. Math. Soc., Providence, RI, 2014.
 
 

\bibitem{D91}
L. A. Dickey: Soliton equations and Hamiltonian systems. \textit{Advanced Series in Mathematical Physics,} 12, (1991). World Scientific Publishing Co., Inc., River Edge, NJ.

\bibitem{GH78}
\corr{P. Griffiths, J. Harris: Principles of algebraic geometry, \textit{Pure and applied mathematics} (1978) A Wiley-Interscience publication, John Wiley \& Sons.}


\bibitem{K17}
Y. Kodama: KP solitons and the Grassmannians. Combinatorics and geometry of two-dimensional wave patterns. \textit{SpringerBriefs in Mathematical Physics,} 22, (2017). Springer, Singapore.

\bibitem{K18}
Y. Kodama: Solitons in two-dimensional shallow water. \textit{SIAM}. 

\bibitem{M79}
V. B. Matveev: Darboux transformation and explicit solutions of the Kadomtcev-Petviaschvily equation, depending on functional parameters. \textit{Lett. Math. Phys.} 3 (1979), no. 3, 213–216. 

\bibitem{M15} 
T. Mizumachi: Stability of line solitons for the KP-II equation in $\mathbb R^2$, \textit{Mem. Amer. Math. Soc.},  238 (2015), no. 1125, vii+95 pp.

\bibitem{MS91}
V. B. Matveev, M. A. Salle: 
Darboux transformations and solitons. \textit{Springer series in nonlinear dynamics}   (1991), Berlin ; New York : Springer-Verlag.

\bibitem{MST11} 
L. Molinet, J. Saut, N. Tzvetkov: Global well-posedness for the KP-II equation on the background of a non-localized solution. \textit{Ann. Inst. H. Poincare Anal. Non Lineaire } \textbf{28}  (2011),  no. 5, 653-676.

\bibitem{NMPZ84}
S. Novikov, S. V. Manakov, L. P. Pitaevskiĭ, V. E. Zakharov: Theory of solitons. The inverse scattering method. \textit{Contemporary Soviet Mathematics, } (1984),  Consultants Bureau [Plenum], New York and London. 


\bibitem{St82}
\corr{S. G. Krantz: Function theory of several complex variables, \textit{Pure and applied mathematics} (1982) A Wiley-Interscience publication, John Wiley \& Sons.}

\bibitem{S76}
J. Satsuma: N-soliton solution of the two-dimensional Korteweg-deVries equation. \textit{J. Phys. Soc. Japan} 40 (1976), no. 1, 286–290.

\bibitem{S79}
J. Satsuma: A Wronskian Representation of N-Soliton Solutions of Nonlinear Evolution Equations. \textit{J. Phys. Soc. Japan} 46 (1979), no. 1, 359-360.

\bibitem{V62} I. N. Vekua, \textit{Generalized analytic functions}, 1962, Pergamon Press, London-Paris-Frankfurt; Addison-Wesley Publishing Co., Inc., Reading, Mass..


\bibitem{VA04} 
J. Villarroel, M. J. Ablowitz: On the initial value problem for the KPII equation with data that do not decay along a line. \textit{Nonlinearity}  \textbf{17}  (2004),  no. 5, 1843-1866.

\bibitem{W87} 
M. V. Wickerhauser: Inverse scattering for the heat operator and evolutions in $2+1$ variables. \textit{Comm. Math. Phys.}  108  (1987),  no. 1, 67-89. 

 
\bibitem{Wu18} 
D. C. Wu: The direct problem for the perturbed  Kadomtsev - Petviashvili II one line solitons. \textit{arXiv:1807.01420}  (2018) 1-33. 

\end{thebibliography}
\end{document}